\pdfoutput=1
\PassOptionsToPackage{bookmarksnumbered,unicode}{hyperref}
\documentclass[sigconf]{acmart}
\AtBeginDocument{%
  \providecommand\BibTeX{{%
    \normalfont B\kern-0.5em{\scshape i\kern-0.25em b}\kern-0.8em\TeX}}}

\usepackage{lipsum}
\usepackage{graphicx}
\usepackage{bm}
\usepackage{amsfonts,amsmath}
\usepackage{multirow}
\usepackage{algorithm}
\usepackage[noend]{algpseudocode}
\usepackage{mathtools}
\usepackage[inline]{enumitem}
\usepackage{pgfplots}
\usetikzlibrary{fit,backgrounds,arrows,decorations.markings,arrows.meta,matrix,positioning,calc,shapes.geometric}
\usepackage[skins]{tcolorbox}
\usepackage{amsthm,bbm}
\makeatletter
\newcommand{\gettikzxy}[3]{%
  \tikz@scan@one@point\pgfutil@firstofone#1\relax
  \edef#2{\the\pgf@x}%
  \edef#3{\the\pgf@y}%
}
\makeatother

\usepackage{subcaption}
\usepackage{ctable}
\usepackage{outlines}

\usepackage[symbol]{footmisc}
\usepgfplotslibrary{groupplots}

\newcommand{\graphact}{{\fontfamily{lmtt}\selectfont GraphACT}}

\DeclarePairedDelimiterX\set[1]\lbrace\rbrace{#1}

\newcommand{\rdsp}{\text{R}_\text{DSP}}
\newcommand{\rbram}{\text{R}_\text{BRAM}}
\newcommand{\rbw}{\text{R}_\text{BW}}
\newcommand{\func}[2][]{
    \ifthenelse{\equal{#2}{}}
    {\text{\fontfamily{lmtt}\selectfont #1}}
    {\text{\fontfamily{lmtt}\selectfont #1}\left(#2\right)}}

\newcommand{\Gg}[1][]{
    \ifthenelse{\equal{#1}{}}
    {\mathcal{G}}
    {\mathcal{G}}_{#1}}
\newcommand{\Vg}[1][]{
    \ifthenelse{\equal{#1}{}}
    {\mathcal{V}}
    {\mathcal{V}_{#1}}}
\newcommand{\Eg}[1][]{
    \ifthenelse{\equal{#1}{}}
    {\mathcal{E}}
    {\mathcal{E}_{#1}}}
\newcommand{\Wg}[1][]{
    \ifthenelse{\equal{#1}{}}
    {\mathcal{W}}
    {\mathcal{W}_{#1}}}
\newcommand{\X}[1][]{
    \ifthenelse{\equal{#1}{}}
    {\bm{X}}
    {\bm{X}^{\paren{#1}}}}
\newcommand{\Anorm}[1][]{
    \ifthenelse{\equal{#1}{}}
    {\widetilde{\bm{A}}}
    {\widetilde{\bm{A}}_{#1}}}

\newcommand{\W}[1][]{
    \ifthenelse{\equal{#1}{}}
    {\bm{W}}
    {\bm{W}^{\paren{#1}}}}

\newcommand{\Mg}[1][]{
    \ifthenelse{\equal{#1}{}}
    {\mathcal{M}}
    {\mathcal{M}}_{#1}}

\newcommand{\trans}{\mathsf{T}}

\newcommand{\size}[1]{\left\lvert #1 \right\rvert}
\newcommand{\paren}[1]{\left( #1 \right)}

\tikzset{
    vertex/.style = {
        circle,
        fill            = black,
        outer sep = 2pt,
        inner sep = 1pt,
    }
}

\newtheorem{theorem}{Theorem}[section]

\copyrightyear{2020}
\acmYear{2020}
\setcopyright{acmcopyright}
\acmConference[FPGA '20]{Proceedings of the 2020 ACM/SIGDA International Symposium on Field-Programmable Gate Arrays}{February 23--25, 2020}{Seaside, CA, USA}
\acmBooktitle{Proceedings of the 2020 ACM/SIGDA International Symposium on Field-Programmable Gate Arrays (FPGA '20), February 23--25, 2020, Seaside, CA, USA}
\acmPrice{15.00}
\acmDOI{10.1145/3373087.3375312}
\acmISBN{978-1-4503-7099-8/20/02}
\begin{document}

\fancyhead{}

\title{{\graphact}: \underline{A}ccelerating G\underline{C}N \underline{T}raining on CPU-FPGA Heterogeneous Platforms}

\author{Hanqing Zeng}
\affiliation{\institution{University of Southern California}
\city{Los Angeles}
\state{California}}
\email{zengh@usc.edu}

\author{Viktor Prasanna}
\affiliation{\institution{University of Southern California}
\city{Los Angeles}
\state{California}}
\email{prasanna@usc.edu}

\begin{abstract}
Graph Convolutional Networks (GCNs) have emerged as the state-of-the-art deep learning model for representation learning on graphs.
It is challenging to accelerate training of GCNs, due to 
\begin{enumerate*}
\item substantial and irregular data communication to propagate information within the graph, and 
\item intensive computation to propagate information along the neural network layers.
\end{enumerate*}
To address these challenges, we design a novel accelerator for training GCNs on CPU-FPGA heterogeneous systems, by incorporating multiple algorithm-architecture co-optimizations. 
We first analyze the computation and communication characteristics of various GCN training algorithms, and select a subgraph-based algorithm that is well suited for hardware execution.
To optimize the feature propagation within subgraphs, we propose a light-weight pre-processing step based on a graph theoretic approach. Such pre-processing performed on the CPU significantly reduces the memory access requirements and the computation to be performed on the FPGA. 
To accelerate the weight update in GCN layers, we propose a systolic array based design for efficient parallelization.
We integrate the above optimizations into a complete hardware pipeline, and analyze its load-balance and resource utilization by accurate performance modeling. 
We evaluate our design on a Xilinx Alveo U200 board hosted by a $40$-core Xeon server. On three large graphs, we achieve an order of magnitude training speedup with negligible accuracy loss, compared with state-of-the-art implementation on a multi-core platform. 
\end{abstract}

\maketitle

\section{Introduction}
\label{sec: intro}

Recently, representation learning on graphs has attracted much attention. 
By extracting structured, low dimensional features from the unstructured, high dimensional graph, many downstream tasks such as node classification \cite{gcn,graphsage}, link prediction \cite{gcn_link}, graph classification \cite{diff-pool} and clustering \cite{mgae} can be performed easily and effectively. 
Among the numerous representation learning methods, Graph Convolutional Networks (GCNs) \cite{gcn} are capable of learning significantly better features than traditional methods \cite{deepwalk,node2vec}. GCNs have been used to facilitate user recommendation for social networks \cite{gcn_web}, to identify protein functionality from interaction graphs \cite{graphsage} and to help circuit testability analysis for EDA \cite{dac_gcn}. 

Despite the popularity of GCNs, training is still expensive in terms of time and computation resources. 
To scale GCNs to larger graphs and to enable fast re-training on dynamic graphs, it is critical to develop accelerators. Existing works \cite{graphsage,ipdps19} parallelize training on GPU and multi-core platform. How to accelerate training by exploiting the FPGA hardware is currently not well studied. 

Similar to Convolutional Neural Network (CNNs), GCNs are built by iteratively stacking multiple layers. Operations of a \emph{graph convolutional layer} are decomposed into two major steps, to 
\begin{enumerate*}
\item propagate information within the graph and 
\item propagate information along the neural network layers.
\end{enumerate*}
Correspondingly, FPGA accelerators need to address the below challenges to improve performance:

\paragraph{Memory access} The feature propagation (step 1) in a large and sparse graph incurs high volume of irregular memory accesses, both on-chip and off-chip. 
The memory challenge is unique to the GCN training problem.
While CNN accelerators  \cite{fpga17_loop,fpga_systolic,caffeine,systolic_freq,Hanqing_FPGA18,dnnbuilder,hanqing_fpl19} achieve high data reuse and regular accesses by tiling and sliding windows on input tensors, such tensor operations do not generalize to sparse graphs. 
In addition, while graph processing accelerators \cite{fpgp,foregraph,shijie_fccm,lijing_graph,lijing_graph2} optimize data flow and layout for propagation of node labels, such optimizations hardly lead to performance gain when the labels are replaced with long feature vectors. 

\paragraph{Computation} The propagation along GCN layers (step 2) involves computationally intensive dense tensor operations on the model weights and graph node features. 
Fast computation of this step require high consumption of hardware resources. 

\paragraph{Load-balance}
Degree-imbalance of graph nodes can significantly degrade the performance of feature propagation. 
Consequently, the overall layer computation (steps 1 and 2) can be load-imbalanced. 
It is challenging to design on-chip computation modules and the corresponding scheduler to ensure load-balance for arbitrary graphs.

We propose {\graphact}, a framework to \underline{a}ccelerate G\underline{C}N \underline{t}raining by addressing the three challenges.
The main contributions are:

\begin{outline}
\1 We propose {\graphact} with the following optimizations:
    \2 \textbf{Algorithm selection}: By analyzing various GCN training algorithms, we select a subgraph-based minibatch algorithm to significantly reduce CPU-FPGA communication. 
    \2 \textbf{Redundancy reduction}: For each subgraph as a minibatch, we identify and eliminate the recurring aggregation operations between common node neighbors. Such online pre-processing on CPU significantly reduces the number of on-chip operations and BRAM accesses on FPGA. 
    \2 \textbf{Parallelization}: We parallelize the key training steps with optimized on-chip computation modules. Integrating the modules into the processing pipeline, we achieve load-balance on a wide range of target FPGA devices. 
\1 We develop an accurate performance model for {\graphact}. The model identifies architectural parameters of the FPGA design, and algorithmic parameters of the minibatch sampler. 
\1 We evaluate {\graphact} on a Xilinx Alveo U200 board hosted by a 40-core Xeon server. We achieve an order of magnitude training time speedup with negligible accuracy loss, compared with state-of-the-art multi-core implementation. 
\end{outline}
\section{Background and Related Work}

We use bold capital letters (e.g., $\X$) to denote matrices (zero indexed). We use $\bm{X}_u$, $\bm{X}_{:,v}$ and $\bm{X}_{u,v}$ to denote a row, a column and an element of $\bm{X}$, respectively. 
We use $\bm{X}\left[a:b,:\right]$ to retrieve a sub-matrix by extracting rows of $\bm{X}$ (from the $a^{\text{th}}$ to the $(b-1)^{\text{th}}$ row). Similarly, we use $\bm{X}\left[:,a:b\right]$ to retrieve a sub-matrix by extracting the columns. 

We use superscript $\trans$ to denote matrix transpose. Superscript within brackets $\paren{\cdot}$ denote index of a GCN layer. 

\subsection{Graph Convolutional Networks}
\label{sec: back gcn}
Graph Convolutional Networks (GCNs) are built
by extending the convolution operation defined on grid matrices to unstructured graphs
\cite{gcn}.
The input to a GCN is an un-directed, node attributed graph $\Gg\paren{\Vg,\Eg,\X}$, where $\Vg$ and $\Eg$ denote the set of all nodes and edges. Each node is attributed by a length-$f$ feature vector, so $\X\in \mathbb{R}^{\size{\Vg}\times f}$ denotes the feature matrix of $\Vg$. 
A GCN embeds each graph node into a $f'$-dimensional vector.
So the output of a GCN is simply an embedding matrix $\X'\in\mathbb{R}^{\size{\Vg}\times f'}$, generated by information from both the node feature and the node connections.

A GCN consists of multiple graph convolutional layers. 
Figure \ref{fig: overview gcn} visualizes layer $\paren{\ell+1}$, where $0\leq\ell\leq L-1$. The layer has $\size{\Vg}$ input nodes, each attributed by a length-$f^{\paren{\ell}}$ vector $\paren{\bm{X}_u^{\paren{\ell}}}^\trans$. 
Note, for the first layer (i.e., $\ell=0$), $\bm{X}^{\paren{0}}=\bm{X}$, and $f^{\paren{0}}=f$. 
For the last layer, $\bm{X}^{\paren{L}}=\bm{X}'$ and $f^{\paren{L}}=f'$.
The layer performs the following operations to obtain the length-$f^{\paren{\ell+1}}$ output vectors $\paren{\bm{X}_u^{\paren{\ell+1}}}^\trans$:

\paragraph{Feature aggregation} 
As shown in Figure \ref{fig: overview gcn}, each node $u$ sends its feature $\paren{\bm{X}_u^{\paren{\ell}}}^\trans$ via two types of connection: the self-connection (blue) and the neighbor-connection (green). Neighbor connection is defined by the edges $\Eg$. 
The features propagated to the same destination via neighbor-connection are aggregated by vector mean. Use node $0$ as an example. 
Its input is $\paren{\bm{X}_0^{\paren{\ell}}}^\trans=\left[2,0,6\right]^\trans$, and it has neighbors $1$, $2$, $3$. This step produces two vectors for node $0$: the blue one  $\paren{\bm{X}_0^{\paren{\ell}}}^\trans$, and the green one $\frac{1}{3} \paren{\bm{X}_1^{\paren{\ell}}+\bm{X}_2^{\paren{\ell}}+\bm{X}_3^{\paren{\ell}}}^\trans$. 

\paragraph{Weight transformation}
Each node independently transforms its two vectors output from the feature aggregation step. In Figure \ref{fig: overview gcn}, we denote the transform function as $h\paren{\cdot}$, parameterized by the GCN weights. 
Transformation function for $u$ and $v$ are identical. 

Equation \ref{eq: gcn forward} defines the forward path of layer $\ell+1$
. 
$\bm{D}$ is the (diagonal) degree matrix, where $\bm{D}_{u,u}$ equals degree of node $u$, and $\bm{D}_{u,v}$ equals zero for $u\neq v$. 
$\bm{A}$ is the adjacency matrix of $\Gg$, where $\bm{A}_{u,v}$ is $1$ if $\paren{u,v}\in\Eg$, and $0$ otherwise. 
Each layer has two weight matrices: the self-weight $\bm{W}_\circ^{\paren{\ell+1}}$ and the neighbor-weight $\bm{W}_\star^{\paren{\ell+1}}$. The ``$|$'' operation concatenates two matrices column-wise. 
The feature aggregation operation corresponds to $\bm{D}^{-1}\cdot \bm{A}\cdot \bm{X}^{\paren{\ell}}$, and the transformation function $h\paren{\cdot}$ applies weights $\bm{W}_\circ^{\paren{\ell+1}}$, $\bm{W}_\star^{\paren{\ell+1}}$, concatenates the intermediate results, and applies $\func[ReLU]{}$ activation.

\begin{equation}
\label{eq: gcn forward}
    \bm{X}^{\paren{\ell+1}}=\func[ReLU]{\bm{X}^{\paren{\ell}}\cdot \bm{W}_\circ^{\paren{\ell+1}}\big\vert \bm{D}^{-1}\cdot \bm{A}\cdot \bm{X}^{\paren{\ell}}\cdot \bm{W}_\star^{\paren{\ell+1}}}
\end{equation}

The output embedding $\bm{X}'$ (or $\bm{X}^{\paren{L}}$) of a GCN is often used in downstream applications such as node classification \cite{graphsage,fastgcn,s-gcn}. 
To classify nodes, we feed $\bm{X}'$ to a Multi-Layer Perceptron (MLP). A softmax layer \cite{deeplearning} converts the MLP outputs to node labels. 

Such node classification procedure enables supervised learning of $\bm{W}_\circ^{\paren{\ell}}$, $\bm{W}_\star^{\paren{\ell}}$. 
Suppose each training node $v\in\Vg$ is provided with a ground truth label. 
During training, we calculate cross-entropy loss $\mathcal{L}$ to measure difference between the ground truth and the generated labels. 
We minimize the loss by updating $\bm{W}_\circ^{\paren{\ell}}$, $\bm{W}_\star^{\paren{\ell}}$ using gradient descent. 
Gradients with respect to layer $\ell$ parameters can be calculated from gradients of layer $\ell+1$ by chain rule:

\begin{subequations}
\label{eq: gcn backward}
  \begin{align}
    \frac{\partial \mathcal{L}}{\partial \bm{W}_\circ^{\paren{\ell}}} =& \paren{\bm{X}^{\paren{\ell-1}}}^{\trans} \cdot \func[mask]{\frac{\partial \mathcal{L}}{\partial \bm{X}^{\paren{\ell}}}\left[:,0:\frac{1}{2}f^{\paren{\ell}}\right]}\label{eq: grad wself}\\
    \frac{\partial \mathcal{L}}{\partial \bm{W}_\star^{\paren{\ell}}} =& \paren{\bm{D}^{-1} \bm{A}\bm{X}^{\paren{\ell-1}}}^{\trans}\cdot \func[mask]{\frac{\partial \mathcal{L}}{\partial \bm{X}^{\paren{\ell}}}\left[:,\frac{1}{2}f^{\paren{\ell}}:f^{\paren{\ell}}\right]}\\
    \frac{\partial \mathcal{L}}{\partial \bm{X}^{\paren{\ell-1}}} =& \func[mask]{\frac{\partial \mathcal{L}}{\partial \bm{X}^{\paren{\ell}}}\left[:,0:\frac{1}{2}f^{\paren{\ell}}\right]}\cdot \paren{\bm{W}_\circ^{\paren{\ell}}}^{\trans}\label{eq: grad X}\\\nonumber
    & + \bm{A} \bm{D}^{-1}\cdot\func[mask]{\frac{\partial \mathcal{L}}{\partial \bm{X}^{\paren{\ell}}}\left[:,\frac{1}{2}f^{\paren{\ell}}:f^{\paren{\ell}}\right]}\cdot \paren{\bm{W}_\star^{\paren{\ell}}}^{\trans}
  \end{align}
\end{subequations}

\noindent where function $\func[mask]{\cdot}$ corresponds to gradient of $\func[ReLU]{\cdot}$. Both $\func[mask]{}$ and $\func[ReLU]{}$ are light-weight elementwise functions on matrices.

\begin{figure}[!htb]
    \centering
    \definecolor{c0}{HTML}{f5f5f5}
\definecolor{c1}{HTML}{e0e0e0}
\definecolor{c2}{HTML}{9e9e9e}
\definecolor{c3}{HTML}{616161}
\definecolor{c4}{HTML}{424242}
\definecolor{cself}{HTML}{483399} 
\definecolor{cneigh}{HTML}{00801a} 

\tikzset{
    archnode/.style={thick,circle,c4,draw=c4,fill=c0},
}
\begin{tikzpicture}[scale=0.95,
outer/.style={draw=gray,dashed,fill=green!1,thick},
>={Stealth[inset=0pt,
length=4pt,angle'=45,round]},scale=1.,
sty edge sel/.style={ultra thick,color=c4},
sty node sel/.style={ultra thick,circle,c4,draw=c4,fill=c0,minimum size=\sizenode},
sty node unsel/.style={circle,c2,draw=c2,fill=c0,minimum size=\sizenode},
sty edge unsel/.style={color=c2},
sty h trans/.style={fill=c0,draw=c4,minimum size=\sizenode, thick},
]

\def\sizenode{8};
\def\xoff{2.};
\node[sty node sel] (0) at (1.5+\xoff,0.7) {0};
\node[sty node sel] (1) at (0+\xoff,0) {1};
\node[sty node sel] (2) at (1.5+\xoff,-.7) {2};
\node[sty node sel] (3) at (3+\xoff,0.1) {3};
\node[sty node sel] (4) at (5+\xoff,0.7) {4};
\node[sty node sel] (5) at (4.5+\xoff,-0.7) {5};

\def\gcnx{1};
\def\gcny{-6};
\def\offsetx{1.3};
\def\offsety{1.3};
\def\sizenodegcn{6};
\def\scalegcn{0.9};

\node[archnode,minimum size=\sizenodegcn,scale=\scalegcn] (g00) at (\gcnx,\gcny) {0};
\node[archnode,minimum size=\sizenodegcn,scale=\scalegcn] (g01) at (\gcnx+\offsetx,\gcny) {1};
\node[archnode,minimum size=\sizenodegcn,scale=\scalegcn] (g02) at (\gcnx+2*\offsetx,\gcny) {2};
\node[archnode,minimum size=\sizenodegcn,scale=\scalegcn] (g03) at (\gcnx+3*\offsetx,\gcny) {3};
\node[archnode,minimum size=\sizenodegcn,scale=\scalegcn] (g04) at (\gcnx+4*\offsetx,\gcny) {4};
\node[archnode,minimum size=\sizenodegcn,scale=\scalegcn] (g05) at (\gcnx+5*\offsetx,\gcny) {5};

\def\belowvec{.07cm};
\node[below=\belowvec of g00.south,align=center] (0 label) {$\paren{\bm{X}_0^{\paren{\ell}}}^\trans$\\$\begin{bmatrix}2\\ 0\\ 6\end{bmatrix}$};
\node[below=\belowvec of g01.south,align=center] (1 label) {$\paren{\bm{X}_1^{\paren{\ell}}}^\trans$\\$\begin{bmatrix}6\\ 2\\ 2\end{bmatrix}$};
\node[below=\belowvec of g02.south,align=center] (2 label) {$\paren{\bm{X}_2^{\paren{\ell}}}^\trans$\\$\begin{bmatrix}2\\ 6\\ 6\end{bmatrix}$};
\node[below=\belowvec of g03.south,align=center] (3 label) {$\paren{\bm{X}_3^{\paren{\ell}}}^\trans$\\$\begin{bmatrix}4\\ 4\\ 4\end{bmatrix}$};
\node[below=\belowvec of g04.south,align=center] (4 label) {$\paren{\bm{X}_4^{\paren{\ell}}}^\trans$\\$\begin{bmatrix}4\\ 0\\ 2\end{bmatrix}$};
\node[below=\belowvec of g05.south,align=center] (5 label) {$\paren{\bm{X}_5^{\paren{\ell}}}^\trans$\\$\begin{bmatrix}0\\ 2\\ 2\end{bmatrix}$};

\node[scale=\scalegcn] (g10) at (\gcnx,\gcny+\offsety) {\textcolor{cself}{$\begin{bmatrix}2\\ 0\\ 6\end{bmatrix}$}\textcolor{cneigh}{$\begin{bmatrix}4\\ 4\\ 4\end{bmatrix}$}};
\node[scale=\scalegcn] (g11) at (\gcnx+\offsetx,\gcny+\offsety) {\textcolor{cself}{$\begin{bmatrix}6\\ 2\\ 2\end{bmatrix}$}\textcolor{cneigh}{$\begin{bmatrix}2\\ 3\\ 6\end{bmatrix}$}};
\node[scale=\scalegcn] (g12) at (\gcnx+2*\offsetx,\gcny+\offsety) {\textcolor{cself}{$\begin{bmatrix}2\\ 6\\ 6\end{bmatrix}$}\textcolor{cneigh}{$\begin{bmatrix}4\\ 2\\ 4\end{bmatrix}$}};
\node[scale=\scalegcn] (g13) at (\gcnx+3*\offsetx,\gcny+\offsety) {\textcolor{cself}{$\begin{bmatrix}4\\ 4\\ 4\end{bmatrix}$}\textcolor{cneigh}{$\begin{bmatrix}2\\ 2\\ 4\end{bmatrix}$}};
\node[scale=\scalegcn] (g14) at (\gcnx+4*\offsetx,\gcny+\offsety) {\textcolor{cself}{$\begin{bmatrix}4\\ 0\\ 2\end{bmatrix}$}\textcolor{cneigh}{$\begin{bmatrix}2\\ 3\\ 3\end{bmatrix}$}};
\node[scale=\scalegcn] (g15) at (\gcnx+5*\offsetx,\gcny+\offsety) {\textcolor{cself}{$\begin{bmatrix}0\\ 2\\ 2\end{bmatrix}$}\textcolor{cneigh}{$\begin{bmatrix}4\\ 2\\ 3\end{bmatrix}$}};

\def\texth{$h\paren{\cdot}$}
\node[sty h trans, scale=\scalegcn] (g20) at (\gcnx,\gcny+2*\offsety) {\texth};
\node[sty h trans, scale=\scalegcn] (g21) at (\gcnx+\offsetx,\gcny+2*\offsety) {\texth};
\node[sty h trans, scale=\scalegcn] (g22) at (\gcnx+2*\offsetx,\gcny+2*\offsety) {\texth};
\node[sty h trans, scale=\scalegcn] (g23) at (\gcnx+3*\offsetx,\gcny+2*\offsety) {\texth};
\node[sty h trans, scale=\scalegcn] (g24) at (\gcnx+4*\offsetx,\gcny+2*\offsety) {\texth};
\node[sty h trans, scale=\scalegcn] (g25) at (\gcnx+5*\offsetx,\gcny+2*\offsety) {\texth};

\foreach \i in {0,...,5}
{
    \node (out\i) at (\gcnx+\i*\offsetx,\gcny+2.8*\offsety) {$\paren{\bm{X}_\i^{\paren{\ell+1}}}^\trans$};
}

\foreach \i in {0,...,5}
{
    \draw [-latex,c4,draw=cself] (g0\i.north) -- (g1\i.south);
}

\foreach \i in {1,2,3}
{
    \draw [-latex,c4,draw=cneigh] (g00.north) -- (g1\i.south);
}
\foreach \i in {0,2}
{
    \draw [-latex,c4,draw=cneigh] (g01.north) -- (g1\i.south);
}
\foreach \i in {0,1,3}
{
    \draw [-latex,c4,draw=cneigh] (g02.north) -- (g1\i.south);
}
\foreach \i in {0,2,4,5}
{
    \draw [-latex,c4,draw=cneigh] (g03.north) -- (g1\i.south);
}
\foreach \i in {3,5}
{
    \draw [-latex,c4,draw=cneigh] (g04.north) -- (g1\i.south);
}
\foreach \i in {3,4}
{
    \draw [-latex,c4,draw=cneigh] (g05.north) -- (g1\i.south);
}

\draw [-latex,c4,fill=c4] (g10.north) -- (g20.south);
\draw [-latex,c4,fill=c4] (g11.north) -- (g21.south);
\draw [-latex,c4,fill=c4] (g12.north) -- (g22.south);
\draw [-latex,c4,fill=c4] (g13.north) -- (g23.south);
\draw [-latex,c4,fill=c4] (g15.north) -- (g25.south);
\draw [-latex,c4,fill=c4] (g14.north) -- (g24.south);

\foreach \i in {0,...,5}
{
    \draw[-latex,c4,fill=c4] (g2\i.north) -- (out\i.south);
}

\path
    (1) [bend left=15] edge [sty edge sel] node [right] {} (0)
    (0) [bend left] edge [sty edge sel] node [right] {} (3)
    (1) [bend right] edge [sty edge sel] node [right] {} (2)
    (2) [bend right] edge [sty edge sel] node [right] {} (3)
    (0) [bend right=10] edge [sty edge sel] node [right] {} (2)
    (3) [bend left] edge [sty edge sel] node [right] {} (4)
    (3) [bend right] edge [sty edge sel] node [right] {} (5)
    (4) [bend right] edge [sty edge sel] node  {} (5);
    
\node[below=.1cm of $(2.south)!0.5!(5.south)$,anchor=north] {Training graph $\Gg\paren{\Vg,\Eg,\bm{X}}$};
\node[below=.1cm of $(2 label.south)!0.5!(3 label.south)$,anchor=north] {One GCN layer based on $\Gg$};

\end{tikzpicture}
    \caption{Overview of the GCN model}
    \label{fig: overview gcn}
\end{figure}
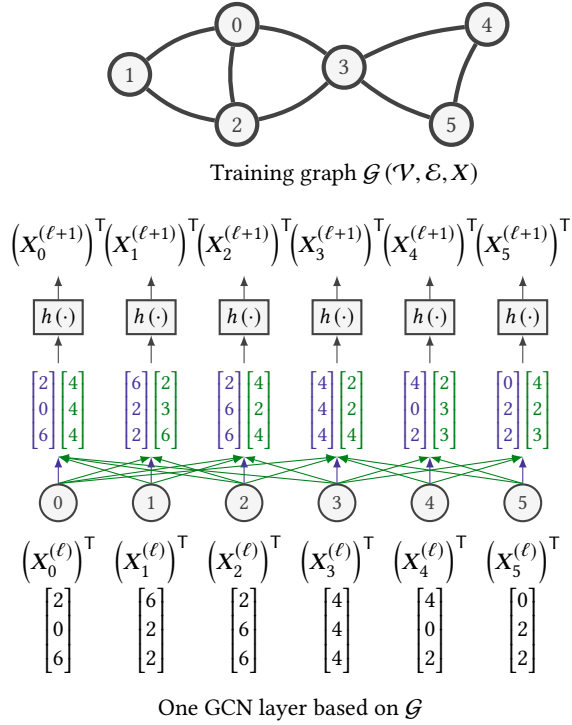

For the node classifier, forward pass of a MLP layer is simply $\bm{X}_\text{MLP}^\text{out}=\func[ReLU]{\bm{W}_\text{MLP}\cdot \bm{X}_\text{MLP}^\text{in}}$, where $\bm{X}_\text{MLP}^\text{in}$ and $\bm{X}_\text{MLP}^\text{out}$ are the input and output features for all $\Vg$, and $\bm{W}_\text{MLP}$ is the layer weight. 
Backward pass of a MLP layer performs computation similar to Equations \ref{eq: grad wself} and \ref{eq: grad X}. 
For the softmax layer, its forward pass involves computation of exponential function, and its backward pass under cross-entropy loss only involves matrix subtraction \cite{grad_softmax}. 

To train GCNs on large graphs, minibatches need to be constructed.
Unlike the case of CNN training where images are independent and identically distributed, training samples of GCNs (i.e., graph nodes) depends on each other due to edges. 
Thus, formulating GCN minibatches is challenging. There have been various techniques to sample minibatches \cite{graphsage,fastgcn,s-gcn,as-gcn,ipdps19,graphsaint}. 
Section \ref{sec: algo selection} analyzes the feasibility of hardware implementation for these techniques.

\subsection{Deep Learning Training Accelerators}

Various FPGA-based accelerators \cite{cnn_asap,cnn_compressed,cnn_cluster,cnn_fpl19,cnn_fpga19} have been proposed to train CNNs. 
The work in \cite{cnn_asap} proposes a modular design based on the types of layer operations, and improves performance via reconfiguration. 
The work in \cite{cnn_cluster} proposes a scalable framework to training CNNs on multi-FPGA clusters. Its partitioning and mapping strategy ensures load-balance. 
The works in \cite{cnn_compressed, cnn_fpga19} accelerate training by model compression. The reduced model size alleviates the burden on BRAM and thus improves resource utilization. 

To accelerate GCN training, the work in \cite{ipdps19} proposes parallelization techniques for multi-core platform. It partitions features to increase cache-hit of each core. 
Although GCNs are an extension of CNNs to graphs, challenges to accelerate GCNs are significantly different. 
GCNs require both sparse and dense matrix operations. Apart from intensive computation, GCN accelerators need to address issues such as irregular memory access and load-balance.

\subsection{Graph Analytics Accelerators}

Many accelerator designs \cite{fpgp,foregraph,shijie_fccm,lijing_graph,lijing_graph2} have been proposed for traditional graph analytic problems, such as PageRank (PR), single source shortest path (SSSP) and breadth first search (BFS). 
The works in \cite{fpgp,shijie_fccm} process input graphs in a two-phased gather-scatter manner, and utilize graph partitioning to improve access locality. 
The work in \cite{foregraph} extends the above memory optimization to multi-FPGA platforms, and the works in \cite{lijing_graph,lijing_graph2} propose optimization specific to HBM and HMC to further boost FPGA performance. 

The memory optimizations proposed in the above works do not directly apply to GCN accelerators. 
First of all, traditional graph analytics often propagate scalars along the graph edges, while GCNs propagate long feature vectors. Thus, although in both cases memory accesses are irregular, the access patterns are very different. 
Secondly, traditional graph analytics often propagate information within the full graph, while GCNs propagate within minibatches. 
Thus, techniques such as graph partitioning may not be effective since minibatch size is much smaller than the full graph size. 

\section{Training Algorithm}

We observe from the forward pass (Equation \ref{eq: gcn forward}) and the backward pass (Equation \ref{eq: gcn backward}), that GCN training involves three types of matrix product $\bm{P}\cdot \bm{Q}$:
\begin{enumerate*}
\item $\bm{P}$ being dense and $\bm{Q}$ being dense,
\item $\bm{P}$ being binary sparse and $\bm{Q}$ being dense, and
\item $\bm{P}$ being diagonal and $\bm{Q}$ being dense.
\end{enumerate*}
From acceleration perspective, type (1) operation is computationally intensive and thus requires massive parallelization of DSPs (Section \ref{sec: arch systolic}); type (2) operation incurs large number of irregular BRAM accesses and motivates the graph topology based optimizations (Section \ref{sec: algo redundancy});
type (3) operation is equivalent to scaling each row of $\bm{Q}$ by the diagonal elements of $\bm{P}$, and thus its contribution to the total training cost is negligible. 
Other operations (e.g., $\func[mask]{}$, $\func[ReLU]{}$, concatenation, sub-matrix extraction) are light-weight and straightforward to implement, and we do not discuss them in detail. 


\subsection{Algorithm Selection}
\label{sec: algo selection}

For efficient training on large graphs, the first step is to reduce external memory communication via minibatches.
Ideally, a minibatch should be sampled such that all data required for gradient calculation (e.g., $\bm{X}^{\paren{\ell}}$ of the minibatch nodes) fits in BRAM.
Among the numerous algorithms \cite{graphsage,fastgcn,s-gcn,as-gcn,ipdps19}, 
some return minibatches not suitable for hardware execution. 
We categorize these algorithms and analyze the hardware cost on external memory accesses.

\paragraph{Minibatch by sampling GCN layers \cite{graphsage,fastgcn,s-gcn,as-gcn}}
Sampling algorithms in this category traverse GCN layers backward from layer-$L$ outputs to layer-$1$ inputs to select minibatch nodes. 
Assume $b$ nodes of layer $\ell+1$ are selected. The sampler then take $\alpha\cdot b$ nodes of layer $\ell$ based on the inter-layer connections and/or the features of the $\overline{d}\cdot b$  neighbors in layer-$\ell$ 
(where $\overline{d}$ is the average node degree).
Specifically, the sampling of \cite{graphsage} does not depend on neighbor features and $10\leq \alpha\leq 50$ in general.
\cite{s-gcn} uses the neighbor features as ``historical activation'' and $\alpha=2$. 
\cite{as-gcn} computes the node probability by the neighbor features and $\alpha=1$ on average. 
\cite{fastgcn} does not require neighbor features and $\alpha=1$.
In summary, suppose we sample $b_0$ output nodes of layer $L$.
Then the sampler reads the neighbor features of each layer (if required) to eventually return $\alpha^L\cdot b_0$ input nodes in layer $1$. 
During training, we read features of the $\alpha^L\cdot b_0$ input nodes of layer $1$, and then perform forward propagation to compute output features of the layer-$\ell$ sampled nodes ($1\leq \ell\leq L$). 




\paragraph{Minibatch by sampling training graph \cite{ipdps19}}
The algorithm samples from $\Gg$ instead of the GCN layers. 
Given a graph sampling algorithm (e.g., multi-dimensional random walk \cite{mrw}), \cite{ipdps19} returns a subgraph $\Gg[s]\paren{\Vg[s],\Eg[s]}$ induced from $\Gg\paren{\Vg,\Eg}$ (where $\size{\Vg[s]}\ll \size{\Vg}$).
The minibatch contains the same set of nodes $\Vg[s]$ for all the GCN layers. In other words, for each minibatch, \cite{ipdps19} constructs a \emph{complete} $L$-layer GCN from $\Gg[s]$, with the layer nodes defined by $\Vg[s]$ and the layer connections defined by $\Eg[s]$. 
Note that unlike \cite{s-gcn,as-gcn}, the graph sampler of \cite{ipdps19} requires no information from node features. 


Suppose we were to implement the above minibatch training algorithms on FPGA. 
Since the full $\bm{X}^{\paren{\ell}}$ is too large to fit on-chip, we need to read from external memory the following data: 
\begin{enumerate*} 
\item layer-$1$ input features of the minibatch nodes, and
\item layer-$\ell$ input features of the minibatch neighbor nodes (if required).
\end{enumerate*}
For ease of analysis, assume the feature sizes of each layer are the same, and let the cost of transferring one feature vector be $1$. 
Also, let the cost of aggregation and computing $h\paren{\cdot}$ for one node (see Figure \ref{fig: overview gcn}) be $1$. Ignore the computation cost of feature aggregation. 
Table \ref{tab: algo sel} summarizes the ratio between on-chip computation cost and off-chip communication cost, where for the ``Value'' row, we set $L=2$, $\overline{d}=15$ and $\alpha$ as $25$, $1$, $2$, $1$ for \cite{graphsage}, \cite{fastgcn}, \cite{s-gcn}, \cite{as-gcn} respectively. 
Algorithms of low computation-communication ratio (i.e., \cite{graphsage,as-gcn,s-gcn}) impede the development of efficient hardware accelerators.
For the remaining two algorithms, the complexity of the sampler of \cite{ipdps19} is much lower than that of \cite{fastgcn}, and the training accuracy of \cite{ipdps19} is higher. 
Thus, we select \cite{ipdps19} as the target training algorithm for acceleration.

\begin{table}[!ht]
\caption{Computation-communication ratio of various training algorithms. Let $B=\sum_{\ell=0}^{L-1}\alpha^{\ell} b_0$, and $B'=L \cdot b_0$. }
\vspace{-.3cm}
    \centering
    \begin{tabular}{rccccc}
        \toprule
        & \cite{graphsage} & \cite{fastgcn} & \cite{s-gcn} & \cite{as-gcn} & \cite{ipdps19}\\
        \midrule
        \midrule
        Expression & $\frac{B}{\alpha^L b_0}$ & $\frac{B}{\alpha^L b_0}$ & $\frac{B}{\alpha^L b_0+\overline{d}B}$ & $\frac{B}{\alpha^L b_0+\overline{d}B}$ & $\frac{B'}{b_0}$ \\
        Value & $0.04$ & $2$ & $0.06$ & $0.06$ & $2$ \\
        \bottomrule
    \end{tabular}
    \label{tab: algo sel}
\end{table}

\paragraph{Remark on notation}
Since we select the algorithm of \cite{ipdps19}, in the following sections, we mainly focus on GCN built on a subgraph $\Gg[s]$. 
A symbol with an ``$s$'' subscript means it is related with a subgraph. 
Thus, forward and backward pass of the GCN can be updated by simply replacing $\bm{A}$, $\bm{D}$, $\bm{X}^{\paren{\cdot}}$ with $\bm{A}_s$, $\bm{D}_s$, $\bm{X}_s^{\paren{\cdot}}$ in Equations \ref{eq: gcn forward} and \ref{eq: gcn backward}. 

\subsection{Redundancy Reduction}
\label{sec: algo redundancy}

Before designing the hardware architecture, we first optimize the redundancy in minibatch training from an algorithm perspective. 
Take Figure \ref{fig: redundancy a} as an example. 
Suppose the graph sampler returns a subgraph $\Gg[s]$ (of $\Gg$ in Figure \ref{fig: overview gcn}) as the minibatch. 
In each GCN layer, nodes $1$, $2$, $3$ and $4$ aggregate neighbor features, with node $0$ calculating $\frac{1}{3}\paren{\bm{X}_1^\trans+\bm{X}_2^\trans+\bm{X}_3^\trans}$, node $1$ calculating $\frac{1}{2} \paren{\bm{X}_0^\trans+\bm{X}_2^\trans}$, node $2$ calculating $\frac{1}{3} \paren{\bm{X}_0^\trans+\bm{X}_1^\trans+\bm{X}_3^\trans}$ and node $3$ calculating $\frac{1}{2} \paren{\bm{X}_0^\trans+\bm{X}_2^\trans}$ (omitting superscript $\paren{\ell}$ for simplicity). 
This corresponds to 6 vector additions and 10 vector reads in total. 
Observe that the vector pair $\paren{\bm{X}_0^\trans, \bm{X}_2^\trans}$ appears in the aggregation of both nodes $1$ and $3$. Similarly, the pair $\paren{\bm{X}_1^\trans,\bm{X}_3^\trans}$ is aggregated by both nodes $0$ and $2$. 
Thus, we can perform pre-processing to compute the partial sum $\bm{X}_0^\trans+\bm{X}_2^\trans$ and $\bm{X}_1^\trans+\bm{X}_3^\trans$. 
The aggregation of the four nodes after pre-processing requires only 2 additions and 4 reads. 
In general, even considering the pre-processing cost, we can still significantly reduce both computation and communication of feature aggregation.
The redundancy reduction helps us achieve perfect load-balance of the \emph{complete pipeline} for a wide range of FPGAs (Section \ref{sec: load balance}). 

\begin{figure}[!htb]
    \centering
    \begin{subfigure}[b]{0.155\textwidth}
    \definecolor{c0}{HTML}{f5f5f5}
\definecolor{c1}{HTML}{e0e0e0}
\definecolor{c2}{HTML}{9e9e9e}
\definecolor{c3}{HTML}{616161}
\definecolor{c4}{HTML}{424242}
\definecolor{cself}{HTML}{1a0080}
\definecolor{cneigh}{HTML}{014d0a}

\tikzset{
    archnode/.style={thick,circle,c4,draw=c4,fill=c0},
}
\begin{tikzpicture}[
outer/.style={draw=gray,dashed,fill=green!1,thick},
>={Stealth[inset=0pt,
length=4pt,angle'=45,round]},scale=1.,
sty edge sel/.style={ultra thick,color=c4},
sty node sel/.style={ultra thick,circle,c4,draw=c4,fill=c0,minimum size=\sizenode},
sty node unsel/.style={circle,c2,draw=c2,fill=c0,minimum size=\sizenode},
sty edge unsel/.style={color=c2},
sty h trans/.style={fill=c0,draw=c4,minimum size=\sizenode, thick},
]

\def\sizenode{10};
\def\xoff{.85};
\node[sty node sel] (0) at (0,0) {0};
\node[sty node sel] (1) at (-\xoff,-\xoff) {1};
\node[sty node sel] (2) at (0,-2*\xoff) {2};
\node[sty node sel] (3) at (\xoff,-\xoff) {3};

\path
    (1) [] edge [sty edge sel] node [right] {} (0)
    (0) [] edge [sty edge sel] node [right] {} (3)
    (1) [] edge [sty edge sel] node [right] {} (2)
    (2) [] edge [sty edge sel] node [right] {} (3)
    (0) [] edge [sty edge sel] node [right] {} (2);

\end{tikzpicture}
    \caption{$\Gg[s]$}
    \label{fig: redundancy a}
    \end{subfigure}
    \begin{subfigure}[b]{0.155\textwidth}
    \definecolor{c0}{HTML}{f5f5f5}
\definecolor{c1}{HTML}{e0e0e0}
\definecolor{c2}{HTML}{9e9e9e}
\definecolor{c3}{HTML}{616161}
\definecolor{c4}{HTML}{424242}
\definecolor{cself}{HTML}{1a0080}
\definecolor{cneigh}{HTML}{014d0a}
\definecolor{cred}{HTML}{dc143c}

\tikzset{
    archnode/.style={thick,circle,c4,draw=c4,fill=c0},
}
\begin{tikzpicture}[
outer/.style={draw=gray,dashed,fill=green!1,thick},
>={Stealth[inset=0pt,
length=4pt,angle'=45,round]},scale=1.,
sty edge sel/.style={ultra thick,color=c4},
sty node sel/.style={ultra thick,circle,c4,draw=c4,fill=c0,minimum size=\sizenode},
sty node unsel/.style={circle,c2,draw=c2,fill=c0,minimum size=\sizenode},
sty edge unsel/.style={color=c2},
sty h trans/.style={fill=c0,draw=c4,minimum size=\sizenode, thick},
]

\def\sizenode{10};
\def\xoff{.85};
\node[sty node sel] (0) at (0,0) {0};
\node[sty node sel] (1) at (-\xoff,-\xoff) {1};
\node[sty node sel] (2) at (0,-2*\xoff) {2};
\node[sty node sel] (3) at (\xoff,-\xoff) {3};

\path
    (0) [] edge [sty edge sel] node [right] {} (2)
    (1) [] edge [sty edge sel] node [right] {} (3);

\end{tikzpicture}
    \caption{$\Gg[a]$ and $\Mg[a]$}
    \label{fig: redundancy b}
    \end{subfigure}
    \begin{subfigure}[b]{0.155\textwidth}
    \definecolor{c0}{HTML}{f5f5f5}
\definecolor{c1}{HTML}{e0e0e0}
\definecolor{c2}{HTML}{9e9e9e}
\definecolor{c3}{HTML}{616161}
\definecolor{c4}{HTML}{424242}
\definecolor{cself}{HTML}{1a0080}
\definecolor{cneigh}{HTML}{014d0a}

\tikzset{
    archnode/.style={thick,circle,c4,draw=c4,fill=c0},
}
\begin{tikzpicture}[
outer/.style={draw=gray,dashed,fill=green!1,thick},
>={Stealth[inset=0pt,
length=4pt,angle'=45,round]},scale=1.,
sty edge sel/.style={ultra thick,color=c4},
sty node sel/.style={ultra thick,circle,c4,draw=c4,fill=c0,minimum size=\sizenode},
sty node sel dashed/.style={dashed,ultra thick,circle,c4,draw=c4,fill=c0,minimum size=\sizenode},
sty node unsel/.style={circle,c2,draw=c2,fill=c0,minimum size=\sizenode},
sty edge unsel/.style={color=c2},
sty h trans/.style={fill=c0,draw=c4,minimum size=\sizenode, thick},
]

\def\sizenode{10};
\def\xoff{.85};
\node[sty node sel dashed] (u) at (0,0) {$u$};
\node[sty node sel] (1) at (-\xoff,-\xoff) {1};
\node[sty node sel] (0) at (-\xoff,-2*\xoff) {0};
\node[sty node sel] (3) at (\xoff,-\xoff) {3};
\node[sty node sel] (2) at (\xoff,-2*\xoff) {2};
\node[sty node sel,dashed] (v) at (0,-\xoff) {$v$};

\path
    (u) [draw,-latex] edge [sty edge sel] node [] {} (1)
    (u) [] edge [sty edge sel] node [right] {} (3)
    (v) [] edge [sty edge sel] node [right] {} (0)
    (v) [] edge [sty edge sel] node [right] {} (2)
    (0) [latex-latex] edge [sty edge sel] node [right] {} (2);

\end{tikzpicture}
    \caption{$\Gg[s]'$}
    \label{fig: redundancy c}
    \end{subfigure}
    \caption{Example on reducing redundancy in feature propagation, using a graph theoretical approach.}
    \label{fig: redundancy abc}
\end{figure}
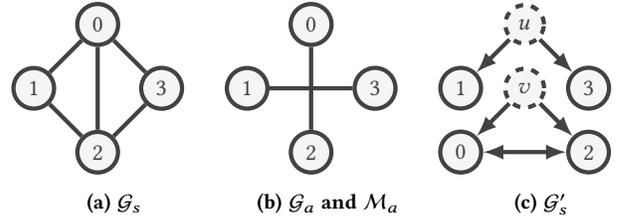


The key to reduce redundancy is to identify sets of nodes appearing frequently in the neighbor lists of $v\in\Vg[s]$. 
To simplify the problem, we aim at finding such sets of size 2 --- we identify \emph{common pairs of neighbors}. 
We first construct an un-directed \emph{aggregation graph} $\Gg[a]$ from $\Gg[s]$ by Algorithm \ref{algo: aggr graph construction}. 
Each edge $\paren{u,v}$ in $\Gg[a]$ represents potential vector sum operations between $u$ and $v$. 
The attribute of $\paren{u,w}$ consists of a set of nodes $\left\{v_1,\hdots, v_n\right\}$, meaning that the neighbor list of $v_i$ ($1\leq i\leq n$) contains both $u$ and $w$.
The weight of an edge is simply the size of its attribute set (we remove edges of weight 1). In Algorithm \ref{algo: aggr graph construction}, $\mathcal{W}_a$ and $\mathcal{D}_a$ can be implemented as a hash table. 
And $\mathcal{W}_a\left[\paren{u,w}\right]$ (or $\mathcal{D}_a\left[\paren{u,w}\right]$) returns the value corresponding to the key $\paren{u,w}$. 
Figure \ref{fig: redundancy b} shows the aggregation graph $\Gg[a]$ for our above example, where edges have weight 2. 

\begin{algorithm}
\caption{Construction of aggregation graph $\Gg[a]$}
\label{algo: aggr graph construction}
\begin{algorithmic}[1]
\renewcommand{\algorithmicrequire}{\textbf{Input:}}
\renewcommand{\algorithmicensure}{\textbf{Output:}}
\Require Subgraph $\Gg[s]\paren{\Vg[s],\Eg[s]}$
\Ensure Aggregator graph $\Gg[a]$
\State $\Vg[a]\gets \Vg[s]$
\State $\mathcal{W}_a\gets \emptyset${\color{blue}\Comment{Key-value pairs mapping edges to edge weights}}
\State $\mathcal{D}_a\gets \emptyset${\color{blue}\Comment{Key-value pairs mapping edges to edge attributes}}
\For {$v\in\Vg[a]$}
    \For {$u\in\func[neigh]{v}$}
        \For {$w\in\func[neigh]{v}\setminus u$}
            \If{Key $\paren{u,w}$ not in $\mathcal{W}_a$}
                \State Add key-value pair $\paren{\paren{u,w},1}$ to $\mathcal{W}_a$
            \Else
                \State $\mathcal{W}_a\left[\paren{u,w}\right]\gets \mathcal{W}_a\left[\paren{u,w}\right]+1$
            \EndIf
            \If{Key $\paren{u,w}$ not in $\mathcal{D}_a$}
                \State Add key-value pair $\paren{\paren{u,w},\left\{v\right\}}$ to $\mathcal{D}_a$
            \Else
                \State $\mathcal{D}_a\left[\paren{u,w}\right]\gets \mathcal{D}_a\left[\paren{u,w}\right]\cup \left\{v\right\}$
            \EndIf
        \EndFor
    \EndFor
\EndFor
\State Remove key-value pairs of weight-1 edges in $\mathcal{W}_a$ and $\mathcal{D}_a$
\State $\Gg[a]\gets$ Un-directed graph based on $\mathcal{W}_a$ and $\mathcal{D}_a$
\end{algorithmic}
\end{algorithm}

Intuitively, the next step upon getting $\Gg[a]$ is to extract the edges with largest weights, so that pre-computation of the corresponding vector sums reduces the most redundancy. 
However, there is one subtlety to notice. Suppose two edges of $\Gg[a]$, $\paren{u_i,u_j}$ and $\paren{u_j,u_k}$, have large weights, and $\left\{v_1,\hdots,v_m\right\}$, the intersection of their attributes,  is non-empty. 
Consider the aggregation of nodes $\left\{v_1,\hdots,v_m\right\}$ 
after pre-computing $\bm{x}'=\bm{X}_i^\trans+\bm{X}_j^\trans$ and $\bm{x}''=\bm{X}_j^\trans+\bm{X}_k^\trans$.
By replacing the pair $\paren{\bm{X}_i^\trans, \bm{X}_j^\trans}$ with $\bm{x}'$, the other pair $\paren{\bm{X}_j^\trans,\bm{X}_k^\trans}$ disappears in aggregation. So $\bm{x}''$ does not help reduce redundancy.
To avoid such useless pre-computation on $\bm{x}''$, a good solution is to find a \emph{maximum weight matching} of $\Gg[a]$, so that the selected pairs imply high redundancy, and share no common nodes.


\begin{theorem}
\label{thm: aggr reduction}
For feature aggregation of each layer, number of vector reads and additions can decrease by at least $\sum_{e\in \Mg[a]^*}\paren{\mathcal{W}_a[e]-2}$ and
$\sum_{e\in\Mg[a]^*}\paren{\mathcal{W}_a[e]-1}$. $\Mg[a]^*$ is maximum weight matching of $\Gg[a]$. 
\end{theorem}

\begin{proof}
We first consider the reduction in number of reads. Since edges in a matching are disjoint, for each $\paren{u,v}\in\Mg[a]^*$, accessing $\bm{X}_u^\trans+\bm{X}_v^\trans$ instead of $\bm{X}_u^\trans$ and $\bm{X}_v^\trans$ saves $\paren{2-1}\cdot \mathcal{W}_a\left[\paren{u,v}\right]$ number of reads.
The pre-computation of $\bm{X}_u^\trans+\bm{X}_v^\trans$ consumes 2 reads. In sum, the total reduction is $\sum\paren{\mathcal{W}_a[e]-2}$.
The proof for reduction in number of additions can be similarly derived.
\end{proof}

Note that Theorem \ref{thm: aggr reduction} considers the total cost including any pre-computation cost. 
Also, although polynomial complexity algorithm \cite{max_matching} exists for computing the maximum weight matching, it is still too expensive in practice. 
Thus, we propose a greedy approach (Algorithm \ref{algo: greedy matching}) to quickly compute a good matching $\Mg[a]$. 

\begin{algorithm}
\caption{Greedy approach to find a good matching $\Mg[a]$}
\label{algo: greedy matching}
\begin{algorithmic}[1]
\renewcommand{\algorithmicrequire}{\textbf{Input:}}
\renewcommand{\algorithmicensure}{\textbf{Output:}}
\Require Aggregation graph $\Gg[a]$; threshold $\theta$
\Ensure Matching $\Mg[a]$
\State $\Mg[a]\gets \emptyset$
\State $\mathbb{H}\gets$ Max-heap of edges $\Eg[a]$, ordered by edge weights
\State $\mathcal{S}\gets$ Length  $\size{\Vg[a]}$ vector of values $\func[True]{}$
\While {$\func[maxWeight]{\mathbb{H}}>\theta$}
    \State $\paren{u,v}\gets \func[extractMaxEdge]{\mathbb{H}}$
    \If {$\lnot\paren{\mathcal{S}_u\land\mathcal{S}_v}$}{\color{blue}\Comment{$\paren{u,v}$ violates edge disjointness}}
        \State \textbf{continue}
    \EndIf
    \State $\mathcal{S}_u\gets \func[False]{};\qquad \mathcal{S}_v\gets \func[False]{}$
    \State $\Mg[a]\gets \Mg[a]\cup \set*{\paren{u,v}}$
\EndWhile
\end{algorithmic}
\end{algorithm}

The next step after obtaining $\Mg[a]$ is to update the original $\Gg[s]$ to $\Gg[s]'$ (following Algorithm \ref{algo: update gs}), so that the pre-computed sums can propagate in $\Gg[s]'$. 
The idea is to merge each pair of nodes in $\Mg[a]$ into a new node, whose feature vector is returned by pre-computation. 
Compared with $\Gg[s]$, the updated $\Gg[s]'$ has more nodes ($\size{\Vg[s]'}=\size{\Vg[s]}+\size{\Mg[a]}$), but less edges. 
Note that $\Gg[s]'$ is directed, even when $\Gg[s]$ is un-directed. The example $\Gg[s]'$ is shown in Figure \ref{fig: redundancy c}.

\begin{algorithm}
\caption{Construction of the update subgraph $\Gg[s]'$}
\label{algo: update gs}
\begin{algorithmic}[1]
\renewcommand{\algorithmicrequire}{\textbf{Input:}}
\renewcommand{\algorithmicensure}{\textbf{Output:}}
\Require Original subgraph $\Gg[s]$; Matching $\Mg[a]$; Edge attribute $\mathcal{D}_a$
\Ensure Updated (directed) subgraph $\Gg[s]'\paren{\Vg[s]',\Eg[s]'}$
\State $\Vg[s]'\gets \Vg[s];\qquad\Eg[s]'\gets \Eg[s]$
\For {$\paren{u,v}\in \Mg[a]$}
    \State Assign a new node $v'$ corresponding to the edge $\paren{u,v}$
    \State $\Vg[s]'\gets \Vg[s]'\cup\set{v'}$
    \For {$w\in \mathcal{D}_a\left[\paren{u,v}\right]$}
        \State Remove $\paren{u,w}$ from $\Eg[s]'$
        \State Replace $\paren{v,w}$ with $\paren{v',w}$ in $\Eg[s]'$
    \EndFor
\EndFor

\end{algorithmic}
\end{algorithm}

\paragraph{Complexity analysis}
Complexity of Algorithms \ref{algo: aggr graph construction}, \ref{algo: greedy matching} and \ref{algo: update gs} are low compared with the feature aggregation. 
Complexity of Algorithm \ref{algo: aggr graph construction} is $\mathcal{O}\paren{\size{\Eg[a]}}=\mathcal{O}\paren{\sum_{v\in\Vg[s]}d_{v}^2}$, where $d_v$ is $v$'s degree. 
Although in the worst case, $\mathcal{O}\paren{\size{\Eg[a]}}=\mathcal{O}\paren{\size{\Vg[s]}\cdot d_\text{max}^2}$, for the graphs in practice, we may assume $\mathcal{O}\paren{\size{\Eg[a]}}=\mathcal{O}\paren{\size{\Vg[s]}\cdot \overline{d}^2}=\mathcal{O}\paren{\size{\Eg[s]}\cdot \overline{d}}$, where $d_\text{max}$ and $\overline{d}$ are the max and average degree of $\Gg[s]$. 
Complexity of Algorithm \ref{algo: greedy matching} is $\mathcal{O}\paren{\size{\Eg[a]}+N\log\size{\Eg[a]}}$, where
the first term counts for line 2, and the second term is for the loop from line 4 to 9. 
Number of times max-edge is extracted from heap, $N$, depends on the threshold $\theta$.
Typically, $N\ll \size{\Eg[a]}$, $\size{\Vg[s]}<5000$, $\overline{d}<20$.
Overall complexity of Algorithms \ref{algo: aggr graph construction} and \ref{algo: greedy matching} is much less than the complexity of single layer feature aggregation (i.e., $\mathcal{O}\paren{\size{\Eg[s]}\cdot f}$, where the feature length $f$ is to the order of $10^2$ to $10^3$). 
For Algorithm \ref{algo: update gs}, lines 6 and 7 each takes at most $d_w$ operations if we simply scan the neighbor list of $w$. In return, we save one vector sum operation for $u$ and $v$. 
Considering $\overline{d}\ll f$, overhead of Algorithm \ref{algo: update gs} is thus negligible compared with the benefit in redundancy reduction. 

It is worth noticing that \begin{enumerate*}
\item one time transformation from $\Gg[s]$ to $\Gg[s]'$ benefits $3L$ number of feature aggregation in a $L$-layer GCN, and 
\item such transformation, which only involves integer operations, reduces floating point arithmetics during feature aggregation. 
\end{enumerate*}
These two observations further justify the cost of Algorithms \ref{algo: aggr graph construction}, \ref{algo: greedy matching} and \ref{algo: update gs}.

\paragraph{Iterative redundancy reduction}
After obtaining $\Gg[s]'$, redundancy still exists when aggregating features of $\Gg[s]'$. 
Viewing $\Gg[s]'$ as the new $\Gg[s]$, we can again apply Algorithms \ref{algo: aggr graph construction}, \ref{algo: greedy matching} and \ref{algo: update gs} to obtain a $\Gg[s]''$. This process can continue until few edges can be reduced (i.e., the matching is small). 
Note that although the training graph (and thus $\Gg[s]$) is often un-directed, Algorithms \ref{algo: aggr graph construction}, \ref{algo: greedy matching} and \ref{algo: update gs} still apply when $\Gg[s]$ is directed. 
Therefore, iterative redundancy reduction is feasible. 

For sake of notation, define one \emph{round} of redundancy reduction as one invocation of Algorithms \ref{algo: aggr graph construction}, \ref{algo: greedy matching} and \ref{algo: update gs}. Define the subgraph output by the last round as $\Gg[s]^\#$, and its adjacency matrix as $\bm{A}_s^\#$. Define the set of matchings in all rounds as $\mathbb{M}_a=\set{\Mg[a], \Mg[a]', \Mg[a]'',\hdots}$.
Define $\gamma_\text{add}\coloneqq\frac{\func[numAdd]{\Gg[s]^\#}+\sum_{\Mg\in\mathbb{M}_a}\size{\Mg}}{\func[numAdd]{\Gg[s]}}$ as the redundancy reduction rate for additions, where $\func[numAdd]{\Gg}= \sum_{v\in \Vg;d_v\geq 1}\paren{d_v-1}$ denotes the number of additions to aggregate by traversing $\Gg$'s neighbor lists. 
Define $\gamma_\text{read}\coloneqq \frac{\func[numRead]{\Gg[s]^\#}+\sum_{\Mg\in\mathbb{M}_a}2\size{\Mg}}{\func[numRead]{\Gg[s]}}$ as the redundancy reduction rate for BRAM reads, where $\func[numRead]{\Gg}=\sum_{v\in\Vg}d_v=\overline{d}\size{\Vg}$ denotes number of reads to aggregate features.  

\section{Accelerator Design}
\label{sec: arch}

\subsection{Overview}
\label{sec: arch overview}

We first partition the workload between FPGA and CPU. 
We let FPGA execute the computation intensive operations, and leave the communication intensive parts to CPU. 
A natural partition is: CPU performs graph sampling, and then converts $\Gg[s]$ to $\Gg[s]^\#$; 
FPGA performs the forward and backward pass defined by Equations \ref{eq: gcn forward} and \ref{eq: gcn backward}, where feature propagation of each layer is based on $\Gg[s]^\#$.
Notice that under supervised learning, the last graph convolutional layer is followed by a node classifier (see Section \ref{sec: back gcn}). The softmax layer at the end of the classifier and the cross-entropy loss require computation of exponential and logarithmic functions. 
Since softmax and loss contribute to a negligible portion of the total workload and their accurate calculation requires complicated hardware \cite{fpga_log, fpga_exp}, we assign their computation to CPU. 
Section \ref{sec: schedule} describes the scheduling.

To improve overall training throughput, we need to
\begin{enumerate*}
\item reduce the overhead in external memory access, and 
\item increase the utilization of the on-chip resources.
\end{enumerate*}
The first goal can be achieved by setting the minibatch parameters so that the size of $\Gg[s]$ is appropriate for the BRAM capacity. 
Ideally, once receiving the initial node features (i.e., $\bm{X}_s^{\paren{0}}$ of $\Vg[s]$) and the connection information (i.e., $\bm{D}_s$, $\bm{A}_s^\#$, $\mathbb{M}_a$), the FPGA should propagate forward from the first graph convolutional layer to the last classifier MLP layer without accessing external memory. 
Similarly, once receiving the gradient with respect to softmax, FPGA should propagate backward without accessing external memory. 
Thus, CPU needs to communicate:

\begin{itemize}
    \item To FPGA: $\bm{X}_s^{\paren{0}}$, $\bm{D}_s$, $\bm{A}_s^{\#}$, $\mathbb{M}_a$ and $\frac{\partial \mathcal{L}}{\partial \bm{X}_\text{MLP}^\text{out}}$
    \item From FPGA: $\bm{X}_\text{MLP}^\text{out}$
\end{itemize}

\noindent And on-chip BRAM needs to store:

\begin{itemize}
    \item Node features: $\bigcup_{0\leq \ell\leq L}\set*{\bm{X}_s^{\paren{\ell}}}$
    \item Subgraph topological data: $\bm{D}_s$, $\bm{A}_s^{\#}$ and $\mathbb{M}_a$
    \item Pre-computed vector sum for pairs in $\mathbb{M}_a$
    \item Intermediate feature aggregation results
    \item Model weights: $\bigcup_{1\leq\ell\leq L} \set*{\bm{W}_\circ^{\paren{\ell}},\bm{W}_\star^{\paren{\ell}}}$ and $\bm{W}_\text{MLP}$
    \item Gradient information: $\frac{\partial \mathcal{L}}{\partial \bm{X}_s^{\paren{\ell}}}$ and $\frac{\partial \mathcal{L}}{\partial \bm{X}_s^{\paren{\ell-1}}}$, for any $2\leq \ell \leq L$
    \item Optimizer specific auxiliary data
    \item Other data (i.e., tile buffer in Section \ref{sec: arch systolic})
\end{itemize}

Note that we only need to store gradients with respect to activations of consecutive layers. When calculating $\frac{\partial \mathcal{L}}{\partial \bm{W}_\circ^{\paren{\ell}}}$ and $\frac{\partial \mathcal{L}}{\partial \bm{W}_\star^{\paren{\ell}}}$, the data $\frac{\partial \mathcal{L}}{\partial \bm{X}_s^{\paren{\ell}}}$ and $\frac{\partial \mathcal{L}}{\partial \bm{X}_s^{\paren{\ell+1}}}$ are stored on chip. When calculating $\frac{\partial \mathcal{L}}{\partial \bm{X}_s^{\paren{\ell-1}}}$, we overwrite $\frac{\partial \mathcal{L}}{\partial \bm{X}_s^{\paren{\ell+1}}}$ with the newly generated $\frac{\partial \mathcal{L}}{\partial \bm{X}_s^{\paren{\ell-1}}}$. 
Also note that the optimizer specific auxiliary data depends on the optimizer used. For vanilla gradient descent, no auxiliary data is needed. 
For gradient descent with momentum \cite{momentum}, we need to store the gradient with respect to model weights for the previous minibatch.
For Adam optimizer \cite{adam}, we need to store more gradient related data (e.g., first moment estimate, second moment estimate, etc. ). 
No matter what optimizer is used, the size of the auxiliary data is comparable to the size of model weights. 

Regarding the processing pipeline on-chip, there are two main computation modules to perform feature aggregation and weight transformation. 
In Figure \ref{fig: overview arch}, the Feature Aggregation module consists of a 1D accumulator array to sum the node vectors. The Weight Transformation module consists of a 2D systolic array to compute the dense matrix product between $\bm{X}$ and $\bm{W}$. 
The two modules are reused for the computation of the $L$ graph convolutional layers, and the Weight Transformation module is also reused for the MLP layer. 
The various BRAM buffers stores the data listed above. 
During forward pass, when computing layer $\ell$, the Feature Aggregation and Weight Transformation modules read from the buffer of $\bm{X}_s^{\paren{\ell-1}}$, and write into $\bm{X}_s^{\paren{\ell}}$. 
In backward pass, for layer $\ell$, the two computation modules read the buffer of $\bm{X}_s^{\paren{\ell-1}}$, and read / write into the two gradient buffers in a ping-pong fashion. 

The $\bm{X}_s^{\paren{\ell}}$ and $\frac{\partial \mathcal{L}}{\partial \bm{X}_s^{\paren{\ell}}}$ buffers use feature major data layout. So each cycle, a buffer can provide the full feature vector of one node. 
For a BRAM block ($36\text{bits}\times 1\text{K}$) of the Xilinx Alveo board we use, we store feature values (32-bit floating point) of $1$K different nodes. 

\begin{figure*}[!htb]
    \centering
    \input{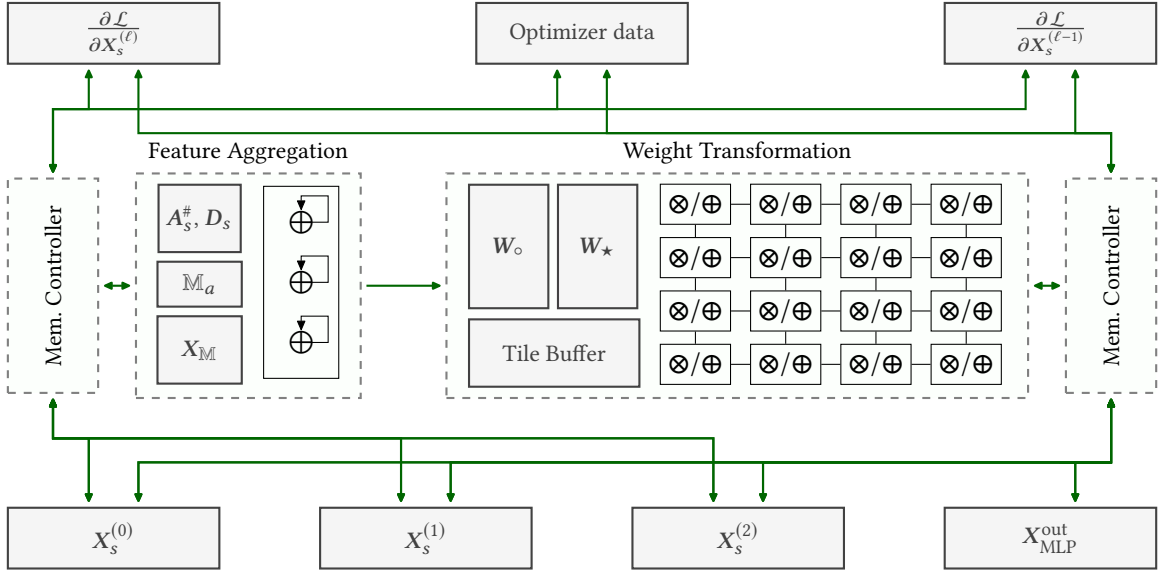}
    \caption{Overview of the processing pipeline on FPGA ($L=2$)} 
    \label{fig: overview arch}
\end{figure*}

\subsection{Feature Aggregation Module}
\label{sec: arch feat aggr}

This module performs feature aggregation in three steps. 
The first pre-computation step calculates the vector sum of node pairs in $\mathbb{M}_a$, and stores the results in the buffer for $\bm{X}_\mathbb{M}$. 
The second step computes $\bm{A}_s^\#\cdot \bm{X}_s^{\paren{\ell}}$ by reading $\bm{X}_\mathbb{M}$ and $\bm{X}_s^{\paren{\ell}}$. 
The final step applies the scaling coefficient based on $\bm{D}_s$. 
The aggregated features are written into a temporary buffer to be consumed by the Weight Transformation module. 
For the below discussion, we ignore step 3 since its complexity is low (i.e., $\mathcal{O}\paren{\size{\Vg[s]}}$) compared with the other steps (i.e., $\mathcal{O}\paren{\size{\Eg[s]}}$).
Regarding steps 1 and 2, since the feature length is large, we explore parallelism along the feature dimension. 
In this case, a simple 1D accumulator array of size $P_\text{agg}=\max_{0\leq \ell \leq L-1} f^{\paren{\ell}}$ is sufficient. 
During pre-computation, pairs of node indices are read sequentially from $\mathbb{M}_a$. Vector sum of each pair consumes two cycles. 
Similarly, during propagation of $\bm{A}_s^\#$, indices in the neighbor lists are read sequentially. 
Note that even though $\mathbb{M}_a$ contains pairs of multiple rounds (see last paragraph of Section \ref{sec: algo redundancy}), as long as the pre-computation on the earlier rounds finishes before that of later rounds, the memory controller for the accumulator array is straightforward. 
For example, assume the original subgraph nodes are indexed continuously from $1$ to $\size{\Vg[s]}$ and the new nodes generated in the first round are indexed from $\size{\Vg[s]}+1$ to $\size{\Vg[s]}+\size{\Mg[a]}$. The matching in the second round, $\Mg[a]'$, can only contain indices from $1$ to $\size{\Vg[s]}+\size{\Mg[a]}$.  
If the second round starts after the first round has finished, all features required by $\Mg[a]'$ are ready. So the accumulator array can directly read $\bm{X}_s^{\paren{\ell}}$ and $\bm{X}_\mathbb{M}$, and continue filling $\bm{X}_\mathbb{M}$ with new features indexed from $\size{\Vg[s]}+\size{\Mg[a]}+1$ to $\size{\Vg[s]}+\size{\Mg[a]}+\size{\Mg[a]'}$.

\paragraph{Remark} 
To further increase parallelism, we can aggregate vectors of multiple neighbors in the same cycle. Then the accumulator array would be replaced by an accumulator tree, and the buffer would be further partitioned to reduce bank conflicts in BRAM. 
In this case, challenges such as load-balance would emerge.
Fortunately, for most target FPGA devices, parallelism of just $\max f^{\paren{\ell}}$ is sufficient. We discuss this in detail in Section \ref{sec: load balance} and \ref{sec: dsp util}. 
Also, we evaluate the storage overhead due to $\bm{X}_\mathbb{M}$ in Section \ref{sec: exp breakdown}.

\subsection{Weight Transformation Module}
\label{sec: arch systolic}

This module performs weight transformation of either a GCN layer (i.e., $h\paren{\cdot}$ function in Figure \ref{fig: overview gcn}) or a MLP layer. 
The main operation is multiplication of dense matrices. 
We use a 2D systolic array to execute the blocked matrix multiplication algorithm. 
Let the dimension of the systolic array be $P_\text{sys}$ (where $P_\text{sys}\ll  f^{\paren{\ell}}$). So the computation parallelism of this module is $P_\text{sys}^2$, and each cycle $2P_\text{sys}$ data are streamed in the array. 
Next we specify the data access pattern. Suppose we compute $\bm{X}\cdot \bm{W}$, where $\bm{X}\in\mathbb{R}^{\size{\Vg[s]}\times f}$ and $\bm{W}\in \mathbb{R}^{f\times f'}$. 
Then the $\bm{X}$ buffer is depth-wise partitioned to tiles of layout $P_\text{sys}\times f$, and the $\bm{W}$ buffer is width-wise partitioned to tiles of layout $f\times P_\text{sys}$. 
There are $\frac{\size{\Vg[s]}}{P_\text{sys}}\times \frac{f'}{P_\text{sys}}$ pairs of tiles. 
For each pair, the systolic array reads a diagonal (or anti-diagonal if matrix is transposed) of the two tiles per cycle. 
Thus, computing a pair of tiles consumes $f+P_\text{sys}-1$ cycles, and computing the full dot product takes $\frac{\size{\Vg[s]}}{P_\text{sys}}\cdot \frac{f'}{P_\text{sys}}\cdot \paren{f+P_\text{sys}-1}\approx \frac{1}{P_\text{sys}^2}\cdot \size{\Vg[s]}\cdot f\cdot f'$ cycles. 
To perform $\func[ReLU]{}$ in the forward pass, we only need to add a comparator in each PE of the systolic array. When the results for a pair of tiles are ready, the PEs clip negative values to zero and append a status bit to indicate whether or not the clipping occurs. 
The $\func[mask]{}$ function in the backward pass can thus be implemented in a simple way based on the $\func[True]{}$ or $\func[False]{}$ of the status bit. 

\paragraph{Interaction with Feature Aggregation module} 
When operating on the neighbor weight $\bm{W}_\star$, Weight Transformation module reads the feature from the buffer filled by the Feature Aggregation module. 
When operating on the self weight $\bm{W}_\circ$, conflicts may occur since both modules read from the $\bm{X}$ buffer. 
We add a small tile buffer in the Weight Transformation module to \emph{completely avoid read conflicts}. 
The tile buffer of size $P_\text{sys}\times f$ stores a tile of $\bm{X}$, and is filled in $f$ cycles. Data in the tile buffer stay for $\paren{f+P_\text{sys}-1}\cdot \frac{f'}{P_\text{sys}}$ cycles to enumerate all tiles of $\bm{W}$. 
Read conflicts can only happen during the filling of the tile buffer. And we simply stall the Feature Aggregation module in this short period of time. 
Since $f'\gg P_\text{sys}$, the pipeline stall has negligible impact on performance. 
The above analysis is based on the forward pass operation. In the backward pass, we swap the dimensions of $\size{\Vg[s]}$ and $f$, and the same conclusion holds.

\subsection{Scheduling}
\label{sec: schedule}


\begin{figure}[!htb]
    \centering
    \definecolor{ccpucomp}{HTML}{00BFFF}
\definecolor{cfpgacomp}{HTML}{90EE90}
\definecolor{ctransfer}{HTML}{5F9EA0}
\definecolor{cfpgacompa}{HTML}{32CD32}
\definecolor{cfpgacompb}{HTML}{ADFF2F}

\begin{tikzpicture}[
outer/.style={draw=gray,dashed,fill=green!1,thick},
labelrot/.style={rotate=90,anchor=center,align=left},
>={Stealth[inset=0pt,
length=4pt,angle'=45,round]},scale=1.,
sty edge sel/.style={ultra thick,color=c4},
sty node sel/.style={ultra thick,circle,c4,draw=c4,fill=c0,minimum size=\sizenode},
sty node unsel/.style={circle,c2,draw=c2,fill=c0,minimum size=\sizenode},
sty edge unsel/.style={color=c2},
]

\def\h{.5cm};
\def\wtrans{.4cm};
\def\wcpusample{1.6cm};
\def\wgap{1.1cm};
\def\wfpga{2.7cm};  
\def\wcpuclass{.3cm};
\node[fill=ctransfer,minimum height=\h,minimum width=\wtrans] (step 1) at (0,0) {1};
\node[fill=ccpucomp,minimum height=\h,minimum width=\wcpusample,left=0cm of step 1.east,anchor=west] (step 7) {7};
\node[fill=ctransfer,minimum height=\h,minimum width=\wtrans, right=\wgap of step 7.east, anchor=west] (step 3) {3};
\node[fill=ccpucomp,minimum height=\h,minimum width=\wcpuclass, right=0cm of step 3.east, anchor=west]
(step 4) {4};
\node[right=0cm of step 4.east,anchor=west,fill=ctransfer,minimum height=\h,minimum width=\wtrans] (step 5) {5};

\node[fill=cfpgacomp,minimum height=\h,minimum width=\wfpga, below=.6cm of step 7.west, anchor=west] (step 2) {2}; 
\node[fill=white,minimum height=\h,minimum width=\wtrans, left=0cm of step 2.west, anchor=east] (step fake) {};
\node[fill=white,minimum height=\h,minimum width=\wtrans, right=0cm of step 2.east, anchor=west] (step fake1) {};
\node[fill=white,minimum height=\h,minimum width=\wcpuclass, right=0cm of step fake1.east,anchor=west] (step fake2) {};
\node[fill=cfpgacomp,minimum height=\h,minimum width=\wfpga,below=.6cm of step 5.east,anchor=west] (step 6) {6};

\node[left=.1cm of step 1.west, anchor=east, align=left] (label cpu) {CPU}; 
\node[left=.1cm of step fake.west, anchor=east,align=left] (label fpga) {FPGA};

\draw[->] ([yshift=-.5cm] step fake.west) --  ([xshift=-.7cm,yshift=-.5cm] step 6.east) node[right] {Time};

\node[below=1cm of step 2.west, anchor=west,fill=cfpgacompa,minimum height=\h, minimum width=0.245*\wfpga] (step a) {a};
\node[left=0cm of step a.east, anchor=west,fill=cfpgacompb,minimum height=\h, minimum width=0.245*\wfpga] (step c) {c};
\node[below=.6cm of step a.west, anchor=west,fill=cfpgacompb,minimum height=\h, minimum width=.245*\wfpga] (step b) {b};
\node[fill=white,minimum height=\h,minimum width=\wcpuclass, right=0cm of step a.west,anchor=east] (step fake a) {};
\node[fill=white,minimum height=\h,minimum width=\wcpuclass, right=0cm of step b.west,anchor=east] (step fake b) {};
\node[right=0cm of step c.east,anchor=west,fill=cfpgacompa,minimum height=\h,minimum width=.245*\wfpga] (step a1) {a};
\node[right=0cm of step a1.east,anchor=west,fill=cfpgacompb,minimum height=\h,minimum width=.245*\wfpga] (step c1) {c};
\node[below=.6cm of step a1.west, anchor=west,fill=cfpgacompb,minimum height=\h, minimum width=.245*\wfpga] (step b1) {b};

\node[below=1cm of step 6.west, anchor=west,fill=cfpgacompa,minimum height=\h, minimum width=0.245*\wfpga] (step aa) {a};
\node[left=0cm of step aa.east, anchor=west,fill=cfpgacompb,minimum height=\h, minimum width=0.245*\wfpga] (step cc) {c};
\node[below=.6cm of step aa.west, anchor=west,fill=cfpgacompb,minimum height=\h, minimum width=.245*\wfpga] (step bb) {b};
\node[right=0cm of step cc.east,anchor=west,fill=cfpgacompa,minimum height=\h,minimum width=.245*\wfpga] (step aa1) {a};
\node[right=0cm of step aa1.east,anchor=west,fill=cfpgacompb,minimum height=\h,minimum width=.245*\wfpga] (step cc1) {c};
\node[below=.6cm of step aa1.west, anchor=west,fill=cfpgacompb,minimum height=\h, minimum width=.245*\wfpga] (step bb1) {b};

\node[left=.1cm of $(step fake a.west)!0.5!(step fake b.west)$, anchor=east, align=left] (label module) {Modules};

\draw (step 1.north) -- ++ (0,.4cm) node[above,labelrot,pos=4.35] {Transfer initial data};
\draw (step 7.north) -- ++ (0,.4cm) node[above,labelrot,pos=4.5] {Sample; pre-process\\(for \textit{next batch})};
\draw ([xshift=.8cm]step 2.north) -- ++ (0,1cm) node[above,labelrot,pos=2.3] {Propagate forward};
\draw (step 3.north) -- ++ (0,.4cm) node[above,labelrot,pos=5.45] {Transfer inputs to softmax};
\draw (step 4.north) -- ++ (0,.4cm) node[above,labelrot,pos=5.5] {Calculate softmax and loss};
\draw (step 5.north) -- ++ (0,.4cm) node[above,labelrot,pos=5.5] {Transfer softmax gradient};
\draw (step 6.north) -- ++ (0,1cm) node[above,labelrot,pos=2.4] {Propagate backward};

\draw (step a.west) -| ($(step a.west)+(-.2cm,0)$) |- ++ (0,-2.1cm) |- ++ (.2cm,0) node[right] {Aggregate features}; 
\draw (step b.south) -- ++ (0,-.8cm) |- ++ (.2cm,0) node[right] {Transform w.r.t. self-weight};
\draw (step c.south) -- ++ (0,-1.cm) |- ++ (.2cm,0) node[right] {Transform w.r.t. neighbor-weight};

\end{tikzpicture}
    \caption{Per-minibatch scheduling between CPU and FPGA (up), and between FPGA computation modules (down). }
    \label{fig: scheduling}
\end{figure}
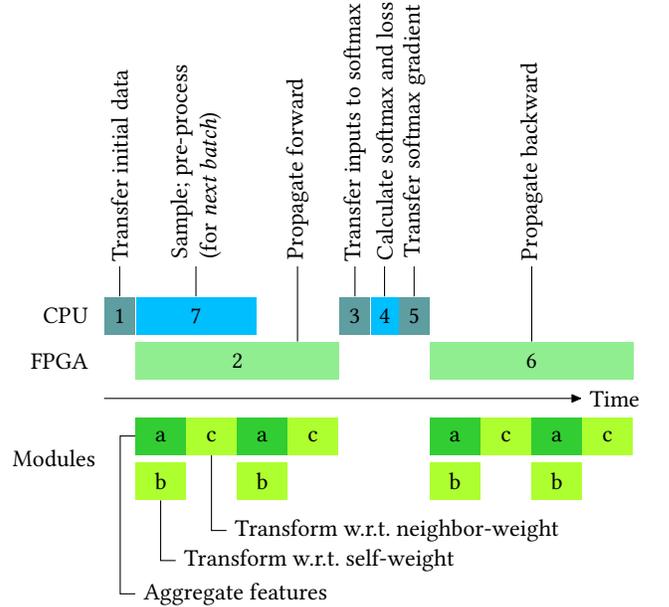

\paragraph{Scheduling between CPU and FPGA}
CPU samples the subgraph $\Gg[s]$, transforms it to $\Gg[s]^\#$, and calculates softmax, loss and the corresponding gradients. These are shown in light blue blocks of Figure \ref{fig: scheduling}. 
FPGA handles majority of the computation workload in the forward and backward pass, as shown in light green. 
Communication between CPU and FPGA, as specified in Section \ref{sec: arch overview}, is in dark blue. 
Notice that the subgraphs are independently sampled for each minibatch. Thus, CPU can prepare $\Gg[s]$ and $\Gg[s]^\#$ for the next minibatch, while simultaneously FPGA is training the current one. This explains the overlap between step 7 and 2. 
Parallelism can be explored by multiple cores of the CPU. The $C$ cores can process subgraphs for the next $C$ minibatches in parallel without any dependency. 

\paragraph{Scheduling of FPGA modules}
In Figure \ref{fig: scheduling}, we only show the scheduling of the GCN forward pass. Scheduling of MLP and the backward pass can be analogously designed. 
Computation on neighbor weights depends on the aggregated feature and computation on self weights has no dependency. 
Thus, we overlap the operations of a and b. 
The avoidance of read conflicts between a and b is discussed in Section \ref{sec: arch systolic}. During operation c, Feature Aggregation module is idle. We analyze its impact on DSP utilization in Section \ref{sec: dsp util}.
\section{Performance Analysis}
\label{sec: perf analysis}

To simplify notation, we assume $f^{\paren{\ell}}=f$, $\forall 0\leq \ell \leq L$. So the weight matrices $\bm{W}_\circ^{\paren{\ell}}\in \mathbb{R}^{f\times \frac{1}{2}f}$ and $\bm{W}_\star^{\paren{\ell}}\in\mathbb{R}^{f\times \frac{1}{2}f}$, where the $\frac{1}{2}$ factor is due to the concatenation in the forward pass. 
For the classifier, we assume a single layer MLP, with $\bm{W}_\text{MLP}\in\mathbb{R}^{f\times f}$.

Regarding FPGA resources, we assume accumulators and multipliers are implemented by DSPs and have the same hardware cost. 
We assume the target FPGA can implement $\rdsp$ number of accumulators / multipliers, store $\rbram$ words, and communicate $\rbw$ words with external memory per cycle. Here a word represents an element of the feature vectors or the weight matrices. 

Training performance depends on the parameters related to:

\begin{itemize}
    \item Minibatch and GCN: $\size{\Vg[s]}$, $\overline{d}$ and $f$
    \item Redundancy reduction: $\gamma_\text{add}$, $\gamma_\text{read}$ and $\sum_{\Mg\in\mathbb{M}_a}\size{\Mg}$
    \item FPGA architecture: $P_\text{agg}$, $P_\text{sys}$
\end{itemize}

We also use the fact that $\overline{d}\ll f\ll \size{\Vg[s]}$ to simplify analysis. 

\subsection{Computation}

Each graph convolutional layer performs feature aggregation 3 times and product on weights 6 times. 
Two of the feature aggregation operate on length-$f$ vectors, and the other one on length-$\frac{1}{2}f$ vectors. 
Since feature aggregation is parallelized along feature dimension only, it takes exactly $\gamma_\text{read}\cdot \size{\Vg[s]}\cdot \overline{d}\cdot f/P_\text{agg}$ cycles to aggregate length-$f$ features. 
All six products on weights have the same complexity, and each takes $\frac{1}{2}\cdot \size{\Vg[s]}\cdot f^2/P_\text{sys}^2$ cycles. By the schedule in Figure \ref{fig: scheduling}, to hide the feature aggregation time, we have:

\begin{subequations}
    \begin{align}
    \gamma_\text{read}\cdot \size{\Vg[s]}\cdot \overline{d}\cdot f\cdot \frac{1}{P_\text{agg}} =& \paren{1-\frac{2P_\text{sys}}{f}}\cdot \frac{1}{2}\cdot \size{\Vg[s]}\cdot f^2\cdot \frac{1}{P_\text{sys}^2}\label{eq: dsp alloc a}\\
    P_\text{agg}+2P_\text{sys}^2 =& \rdsp\label{eq: dsp alloc b}\\
    P_\text{agg} \leq& f\label{eq: dsp alloc c}
    \end{align}
    \label{eq: dsp alloc}
\end{subequations}

\noindent where factor $1-\frac{2P_\text{sys}}{f}$ is due to the pipeline stall analyzed in Section \ref{sec: arch systolic}.
Solving Equations \ref{eq: dsp alloc a} and \ref{eq: dsp alloc b} under the constraint of \ref{eq: dsp alloc c} gives the architectural parameters $P_\text{agg}^*$, $P_\text{sys}^*$. 
Under reasonable values of $f$, $\size{\Vg[s]}$, $\overline{d}$, $\gamma_\text{read}$ and $\rdsp$, the solutions $P_\text{arr}^*$ and $P_\text{sys}^*$ always exist (see also Section \ref{sec: load balance}). Total FPGA cycles\footnote{We ignore the number of cycles to apply gradients to weights (elementwise operation), since it is negligible compared to the number of cycles to calculate gradients.} to complete one minibatch is:

\begin{equation}
    T_\text{batch}=\paren{3L+3}\cdot \size{\Vg[s]}\cdot f^2\cdot \frac{1}{\paren{P_\text{sys}^*}^2}
\end{equation}

\subsection{Communication}
\label{sec: perf mem}
\paragraph{On-chip storage} All data listed in Section \ref{sec: arch overview} have to fit on-chip. Index data $\bm{A}_s^\#$ and $\mathbb{M}_a$ and coefficient $\bm{D}$ are negligible compared with feature data (since $\overline{d}\ll f$). 
Size of the buffer for $\bm{X}_\mathbb{M}$ is $f\cdot \sum_{\Mg\in\mathbb{M}_a}\size{\Mg}$. 
Size of the buffer between Feature Aggregation and Weight Transformation modules is $f\cdot\size{\Vg[s]}$. 
Size of the tile buffer in Weight Transformation module is $P_\text{sys}\cdot\size{\Vg[s]}$. 
Weights for each layer takes $f^2$. Under gradient descent with momentum, the optimizer requires additional $f^2$ storage per layer for the auxiliary data. 
Thus,

\begin{equation}
\paren{L+5} f\size{\Vg[s]} + f \sum\limits_{\Mg\in\mathbb{M}_a}\size{\Mg} + P_\text{sys}\size{\Vg[s]} + 2\paren{L+1}f^2 \leq \rbram
\label{eq: constraint bram}
\end{equation}

\noindent By adjusting $\theta$ and number of rounds (see Algorithm \ref{algo: greedy matching} and Section \ref{sec: algo redundancy}), we control $\sum\size{\Mg}$ to meet a pre-defined budget (say $\size{\Vg[s]}$). Then, we tune the graph sampler so that $\size{\Vg[s]}$ satisfies inequality \ref{eq: constraint bram}.






\paragraph{Off-chip accesses}
Ignoring $\bm{D}_s$, $\bm{A}_s^\#$ and $\mathbb{M}_a$, CPU-FPGA data transfer include the initial features and the MLP output features and gradients. Thus, $\beta$, the ratio between on-chip computation and off-chip communication, is lower bounded\footnote{To simplify expression, we only consider the computation of the systolic array on graph convolutional layers. Therefore, $f$ is just a lower bound on the ratio.} by $f$, indicating high reuse of on-chip data, and low overhead in external memory access. 

\subsection{Load-Balance}
\label{sec: load balance}

If feature aggregation is parallelized along feature dimension only, the full FPGA pipeline is \emph{perfectly} load-balanced, regardless of the graph connection. 
However, if we would have to aggregate features of multiple nodes in parallel, BRAM access conflicts would cause load-imbalance of the module and degrade training performance. 

Load-imbalance is more likely to happen on FPGAs with more DSP resources. The threshold $\hat{\text{R}}_\text{DSP}$ for a design to become load-imbalanced can be derived by plugging $P_\text{agg}=f$ into Equations \ref{eq: dsp alloc a} and \ref{eq: dsp alloc b}. After simplifying the expression, we have:

\begin{equation}
    \hat{\text{R}}_\text{DSP} >  \paren{\frac{f}{\overline{d}\gamma_\text{read}}}^2\cdot \paren{1+\overline{d}\gamma_\text{read}-\sqrt{1+2\gamma_\text{read}\overline{d}}}
\end{equation}

\noindent Feature dimension is often set to $256$, $512$ or $1024$ in the literature. Subgraph degree is often less than $20$, so we set $\overline{d}=15$. The ratio $\gamma_\text{read}$ depends on $\Gg[s]$, and we set it to $0.7$. 
With such parameter values, $\hat{\text{R}}_\text{DSP}>4048$, meaning that there have to be at least $4048$ multipliers / accumulators on chip for our FPGA pipeline to become load-imbalanced.  
Even the largest FPGAs in the market (e.g., Xilinx UltraScale+, Intel Stratix-10) does not have such high DSP capacity (considering that data in single precision floating point are necessary to achieve high training accuracy). 
Note that $\gamma_\text{read}$ helps keep the threshold $\hat{\text{R}}_\text{DSP}$ high. If $\gamma_\text{read}=1$, $\hat{\text{R}}_\text{DSP}$ is only $3038$. 


\subsection{DSP Utilization}
\label{sec: dsp util}


Since load-balance is achieved and BRAM conflicts are eliminated, we can analytically derive the DSP utilization for any $\Gg[s]$. 
From the timeline of Figure \ref{fig: scheduling}, the systolic array is idle when CPU is communicating with FPGA (ignoring the CPU computation time on softmax and loss). The accumulator array of Feature Aggregation module is idle during the stall and during computation on neighbor weights. Using the fact that computation-communication ratio $\beta>f$ (see Section \ref{sec: perf mem}), we derive the overall DSP utilization $\mu$ as:

\begin{equation}
    \mu > \mu'\cdot \frac{1}{1+\mu'\cdot \frac{\rdsp}{f\cdot \rbw}}
\end{equation}

\noindent where $\mu' = \frac{1}{\rdsp}\cdot \paren{2\paren{P_\text{sys}^*}^2 + \frac{5}{12} P_\text{arr}^* \paren{1-\frac{2P_\text{sys}^*}{f}}}$, and $P_\text{arr}^*$, $P_\text{sys}^*$ satisfiy Equations \ref{eq: dsp alloc a} and \ref{eq: dsp alloc b}. The $\frac{5}{12}$ factor is due to the aggregation of length-$\frac{1}{2}f$ features in the backward pass. 
If we can reduce the redundancy up to 
$\gamma_\text{read}>1-\frac{2P_\text{sys}^*}{f}>1-\frac{\sqrt{2\rdsp}}{f}$, we can then lower-bound $\mu'$ by $\frac{1}{1+\overline{d}/f}$. With the typical parameter values specified in Section \ref{sec: load balance}, $\mu'>95\%$. 
For most FPGA platforms, $\frac{\rdsp}{f\cdot \rbw}\ll 1$, so overall DSP utilization $\mu$ is close to $1$. 
For example, on Xilinx Alveo U200 board connected to CPU via PCIe 3.0 x16, we have $\mu>93\%$.
\section{Experiments}
We evaluate on a
Xilinx Alveo U200 board hosted by a $40$-core Xeon server (E5-2698 v4 @2.2GHz, hyper-threaded). 
CPU-FPGA communication is via PCIe 3.0 x16 (peak bandwidth: $15.8$GB/s). 
The UltraScale+ XCU200 FPGA on Alveo has 5867 DSP slices, 1766 36Kb block RAMs and 800 288Kb Ultra RAMs\footnote{In this paper, we refer to block RAM and Ultra RAM by a general term ``BRAM''.}. 
For Float32, on-chip RAM can store 8166K words, and DSPs can implement 1173 accumulators or 1955 multipliers. 
We use Vivado 2018.3 for synthesis.

\begin{table}[!ht]
\caption{Dataset statistics}
    \centering
    \resizebox{\columnwidth}{!}{
    \begin{tabular}{rccccc}
        \toprule
        Dataset & Nodes & Edges & Attribute & Classes & Train/Val/Test\\
        \midrule
        \midrule
        PPI & 14,755 & 225,270 & 50 & 121 & 0.66/0.12/0.22\\
        Reddit & 232,965 & 11,606,919 & 602 & 41 & 0.66/0.10/0.24\\
        Yelp & 716,847 & 6,977,410 & 300 & 100 & 0.75/0.15/0.10\\
        \bottomrule
    \end{tabular}}
    \label{tab: dataset}
\end{table}

\begin{table*}[!ht]
\caption{Comparison of convergence and test set accuracy (F1-micro score)}
    \centering
    \begin{tabular}{rccccccccc}
        \toprule
        \multirow{2}{*}{Method} & \multicolumn{3}{c}{PPI}  & \multicolumn{3}{c}{Reddit} & \multicolumn{3}{c}{Yelp}\\
        \cmidrule(lr){2-4}\cmidrule(lr){5-7}\cmidrule(lr){8-10}& Accuracy & Epochs & Subgraph size & Accuracy & Epochs & Subgraph size  & Accuracy & Epochs & Subgraph size\\
        \midrule
        \midrule
        GraphSAGE \cite{graphsage} & 0.603 & 71 & -- & 0.953 & 6 & -- & 0.593 & 10 & --\\
        Para-GCN \cite{ipdps19} & 0.685 & 120 & 5500 & 0.957 & 5 & 8000 & 0.606 & 8 & 4000 \\
        {\graphact} & 0.678 & 104 & 4000 & 0.952 & 5 & 2750 & 0.598 & 7 & 2750 \\
        \bottomrule
    \end{tabular}
    \label{tab: exp acc}
\end{table*}

\begin{table*}[!ht]
\caption{Comparison of resources and training time with state-of-the-art implementations}
    \centering
    \resizebox{\textwidth}{!}{
    \begin{tabular}{rccccccccc}
        \toprule
        & \multicolumn{3}{c}{Xilinx Alveo U200} &  \multicolumn{3}{c}{Xeon E5-2698 v4 server} & \multicolumn{3}{c}{NVIDIA Tesla P100}\\
        
          \cmidrule(lr){2-4}\cmidrule(lr){5-7}\cmidrule(lr){8-10}  & PPI & Reddit & Yelp & PPI & Reddit & Yelp & PPI & Reddit & Yelp\\
        \midrule
        \midrule
        Data type & Float32 & Float32 & Float32 & Float32 & Float32 &Float32 & Float32 & Float 32 & Float32\\
        Frequency (GHz) & 0.2 & 0.2 & 0.2 & 2.2 & 2.2 & 2.2 & 1.2 & 1.2 & 1.2\\
        DSP slices / CPU cores / CUDA cores & 5632 (96\%) & 5632 (96\%) & 5632 (96\%) & 40 & 40 & 40 & 3584 & 3584 & 3584\\
        BRAM / L3 / HBM2 (MB) & 34.6 (96\%) & 32.8 (91\%) & 27.3 (75\%) & 102 & 102 & 102 & 16280 & 16280 & 16280\\
        Off-chip bandwidth (GB/sec) & 15.8 & 15.8 & 15.8 & 136.6 & 136.6 & 136.6 & 15.8 & 15.8 & 15.8\\
        \midrule
        Total convergence time (sec) & 9.6 & 7.6 & 23.4 & 151.4 & 95.5 & 359.4 & 10.6 & 11.4 & 30.4\\
        \bottomrule
    \end{tabular}}
    \label{tab: exp comparison}
\end{table*}



We use three standard, large benchmark graphs \cite{graphsage,s-gcn,fastgcn,as-gcn,ipdps19} to evaluate the training speed and accuracy.
In Table \ref{tab: dataset}, ``Attribute'' specifies the initial feature length (i.e., $f^{\paren{0}}=\size{\bm{X}_v^{\paren{0}}}$). ``Classes'' specifies total number of node classes. ``Train/Val/Test'' splitting follows \cite{ipdps19}.
For all the datasets, following \cite{graphsage,s-gcn,as-gcn,fastgcn,ipdps19}, the GCN has two graph convolutional layers ($L=2$) and one MLP layer in the classifier. 
The hidden dimension is $f^{\paren{\ell}}=256$, where $\ell=1,2$. 

We compare with the baseline \cite{ipdps19} since {\graphact} uses the same minibatch algorithm as \cite{ipdps19}. 
For \cite{ipdps19}, we run the public implementation\footnote{Baseline code: \url{https://github.com/ZimpleX/gcn-ipdps19}, commit a6f531.} by Python (version 3.6.7) and Tensorflow (version 1.12). 
We run \cite{ipdps19} on both the CPU and GPU platforms. Table \ref{tab: exp comparison} shows the hardware specification of the hardware. 
Parallelization of \cite{ipdps19} on CPU is via the Intel MKL library. Thus, we add the \verb|--mkl| flag to run \cite{ipdps19} on CPU. 
Parallelization of \cite{ipdps19} on GPU is directly via Tensorflow. 
The implementation of {\graphact} is open-sourced\footnote{{\graphact} code: \url{https://github.com/GraphSAINT/GraphACT}}. 

\subsection{Training Accuracy}

Under the same hyper-parameters, the accuracy and convergence of {\graphact} are identical \cite{ipdps19}. 
However, since we set the subgraph size under the BRAM constraint (Equation \ref{eq: constraint bram}), it is necessary to evaluate accuracy under our settings. 
Table \ref{tab: exp acc} summarizes the comparison, where except the subgraph size, all hyper-parameters of Para-GCN and {\graphact} are the same. The popular GraphSAGE does not sample subgraphs, and we use its default minibatch size of 512.  
Comparing with \cite{graphsage}, we achieve much higher accuracy on PPI, and comparable accuracy on Reddit and Yelp after similar number of epochs. 
Comparing with \cite{ipdps19}, reducing the subgraph size $\size{\Vg[s]}$ has negligible impact on the overall accuracy and convergence rate. 

\subsection{Redundancy Reduction}
\label{sec: exp breakdown}

After transforming $\Gg[s]$ to $\Gg[s]^{\#}$, redundancy in computation and communication are reduced at the cost of extra on-chip buffer to store the vector sums of $\mathbb{M}_a$. 
Figure \ref{fig: exp fpga dsp} shows such tradeoff, as well as the effectiveness of redundancy reduction. Each ``$\times$'' marker corresponds to one round of redundancy reduction, and we run for 5 rounds ($\theta=2$, Algorithm \ref{algo: greedy matching}).   
We observe
\begin{enumerate*}
\item number of reads and additions can be reduced to as low as $60\%$, at the cost of a buffer less than $2\size{\Vg[s]}$, 
\item the $\gamma_\text{read}$ value used in Section \ref{sec: load balance} and \ref{sec: dsp util} (to analyze load-balance and DSP utilization) is realistic, and 
\item amount of redundancy depends on the graph topology --- while our algorithm is effective on many graphs (e.g., PPI, Yelp), exceptions exist (e.g., Reddit).
\end{enumerate*}
By Figure \ref{fig: exp fpga dsp}, in our implementation, we set a budget $B$ such that $\sum_{\Mg\in\mathbb{M}_a}\size{\Mg}<B\cdot\size{\Vg[s]}$. For PPI, Reddit, Yelp, $B=2$, $0.5$, $0.5$.

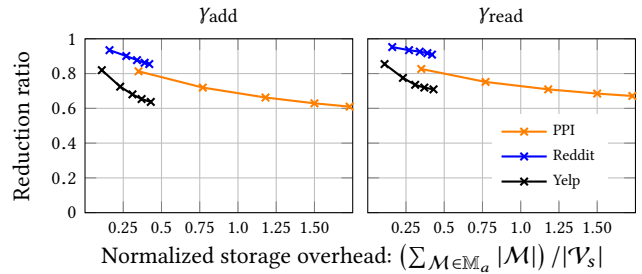
\begin{figure}[!htb]
    \centering
    \begin{tikzpicture}

\begin{groupplot}[
    group style={
    group size=2 by 1,
    y descriptions at=edge left,
    horizontal sep=2mm},
    scale only axis,
    height=0.13\textwidth,
    ymin=0,ymax=1,xmin=0,xmax=1.75,
    tick label style={font=\footnotesize},
    title style={at={(axis description cs:0.5, 0.95)}},
    grid,
    ylabel={Reduction ratio},
    y label style={at={(axis description cs:0.17,.5)},anchor=south},
    xtick={0.25,0.5,0.75,1.0,1.25,1.5},
    xticklabels={0.25,0.5,0.75,1.0,1.25,1.50},
]
\nextgroupplot[
    title={$\gamma_\text{add}$},
    width=0.2\textwidth,
    xlabel={Normalized storage overhead: $\paren{\sum_{\Mg\in\mathbb{M}_a}\size{\Mg}}/\size{\Vg[s]}$},
    every axis x label/.append style={at=(ticklabel cs:1.0)},
    xlabel shift=-17 pt,
]
\addplot[thick,color=orange,mark=x] coordinates{(0.35,1/1.23) (0.77,1/1.39) (1.18,1/1.51) (1.50,1/1.59) (1.73,1/1.64)};
\addplot[thick,color=blue,mark=x] coordinates{(0.16,1/1.07) (0.27,1/1.11) (0.34,1/1.14) (0.39,1/1.16) (0.42,1/1.17)};
\addplot[thick,color=black,mark=x] coordinates{(0.11,1/1.22) (0.23,1/1.38) (0.31,1/1.47) (0.37,1/1.53) (0.43,1/1.57)};

\nextgroupplot[
    title={$\gamma_\text{read}$},
    width=0.2\textwidth,
    legend cell align=left,
    legend style={at={(0.9,0.57)},draw=none},
    legend style={font=\scriptsize}]
\addplot[thick,color=orange,mark=x]coordinates{(0.35,1/1.21) (0.77,1/1.33) (1.18,1/1.41) (1.50,1/1.46) (1.73,1/1.49)};
\addplot[thick,color=blue,mark=x] coordinates{(0.16,1/1.05) (0.27,1/1.07) (0.34,1/1.08) (0.39,1/1.09) (0.42,1/1.10)};
\addplot[thick,color=black,mark=x] coordinates{(0.11,1/1.17) (0.23,1/1.29) (0.31,1/1.36) (0.37,1/1.39) (0.43,1/1.41)};
\legend{PPI, Reddit, Yelp}

\end{groupplot}

\end{tikzpicture}
    \caption{Reduced redundancy \emph{vs.} Storage overhead}
    \label{fig: exp fpga dsp}
\end{figure}

\subsection{Comparison with State-of-the-Art}

Table \ref{tab: exp comparison} shows the training speed comparison with the state-of-the-art \cite{ipdps19}.
All implementations use single precision floating point for weights and features. 
In our design, the DSP and BRAM resources are heavily utilized. We set $P_\text{sys}=24$ and $P_\text{arr}=128$ for all datasets. 
For the multi-core implementation, ``L3'' shows the aggregated L3-cache size and ``Off-chip bandwidth'' shows the peak CPU-main memory data transfer speed. 
For the GPU implementation, ``Off-chip bandwidth'' is the same as the proposed implementation since the GPU is connected via PCIe 3.0 x16. 
``Total convergence time'' includes the time for graph sampling, pre-processing (for redundancy reduction), GCN forward/backward propagation and data transfer from/to main memory. It excludes the time for initial data loading from the disk and Tensorflow initialization (for allocating device resources and building computation graphs). 

Comparing with the CPU baseline, we achieve $12\times$ to $15\times$ speedup. 
Apart from the overhead of the Python language, inefficiency of the CPU baseline is mainly due subgraph feature aggregation operation. 
We observe that although feature aggregation requires less than 10\% of the multiplication/addition of weight transformation, the CPU spends about equal amount of time on the two operations (see also Figure 3.D of \cite{ipdps19}). 
Comparing with the GPU baseline, our design convergences slightly faster ($1.1\times$ to $1.5\times$). 
The theoretical peak performance of the GPU (9.3 TFLOPS for Float32 \cite{p100-spec}) is much higher than that of the FPGA (1.8 TFLOPS for Float32 \cite{u200-spec}\footnote{Each Float32 multiplier consumes 3 DSPs. The max DSP frequency is 775MHz.}). 
Inefficiency of the GPU baseline is mainly due to the sub-optimal Tensorflow implementation of sparse matrix multiplication. 
Note that the CPU in our proposed design only executes light weight operations (Section \ref{sec: algo redundancy} and \ref{sec: schedule}). Thus, the training time improvement mainly comes from the highly optimized FPGA architecture. 

Although the baseline training time may be further improved by re-implementation using C++/CUDA, inefficiency in CPU and GPU due to sparse feature aggregation may not be easily eliminated. 
\section{Discussion}

This work proposed several hardware-software optimizations for GCN training on CPU-FPGA heterogeneous platforms. We discuss these optimizations and their applicability to various platforms. 

\paragraph{Design challenges} The key to accelerating GCN training is to address the challenges of memory access and load-balance. 
The solutions for these differ on various platforms. 
Regarding memory access, the solution on FPGA must optimize both on-chip and off-chip accesses. 
For data in BRAMs, we need to increase their reuse so as to reduce off-chip communication.
We also need to reduce bank access conflicts to reduce pipeline stalls. 
In {\graphact}, we reduce off-chip communication by setting the subgraph size based on the BRAM capacity (Section \ref{sec: perf mem}). We eliminate on-chip access conflicts by appropriately parallelizing the feature aggregation operation (Section \ref{sec: arch feat aggr}) and buffering data tiles between computation modules (Section \ref{sec: arch systolic}).
On the other hand, for CPU and GPU, the critical issue is to enhance data access locality since the memory
hierarchy is not explicitly exposed to the programmer. 
Considering that the full graph fits in the CPU main memory or GPU global memory (sizes of which range from $10^1$ to $10^3$ GB), the memory access challenge on CPU and GPU may be addressed by software node-reordering or by hyper-threading.
Regarding load-balance, CPU and GPU can decouple the operations of feature
propagation and weight transformation, and then adopt separate
strategies to balance them. However, on FPGA, since the two operations are processed
by a single pipeline, an integrated strategy needs to simultaneously balance the very different operations (one on dense tensors and the other on sparse
tensors).
While memory access and load-balance are challenging to optimize, a carefully designed FPGA pipeline can achieve much higher efficiency than the CPU or GPU solutions. This is due to the ability of FPGA to customize computation modules and to control memory accesses in an explicit and fine-grained manner.

\paragraph{Remark on redundancy reduction} Although redundancy reduction is an algorithm-level optimization independent of the platform, it may not be directly applicable to CPU to improve performance. 
Recall that redundancy reduction (Section \ref{sec: algo redundancy}) constructs a subgraph $\mathcal{G}_s^{\#}$ with less edges yet \emph{more} nodes than $\mathcal{G}_s$. 
Thus, feature aggregation on $\mathcal{G}_s^{\#}$ may have worse data locality than $\mathcal{G}_s$. 
On a CPU, even though the L3-cache can hold the node features of $\mathcal{G}_s^{\#}$, the overall feature aggregation performance may not benefit from the reduced computation load due to the potential increase in L1- or L2-cache misses.
On the other hand, data access locality of {\graphact} does not affect FPGA performance, as long as data of one minibatch completely fits in BRAM. 
Under the proposed architecture and data layout (Section \ref{sec: arch}), BRAM will always provide the requested node feature of $\mathcal{G}_s^{\#}$ to the Feature Aggregation module within one cycle.

\section{Conclusion}


We accelerated GCN training on CPU-FPGA heterogeneous platforms. 
By multiple software-hardware co-optimizations, we achieved conflict free BRAM access, load-balance and high DSP utilization. 


We plan to extend {\graphact} to accelerate inference, where the GCN operates on large, un-sampled graphs. 
The redundancy reduction and the computation modules of {\graphact} can be applied to achieve high DSP utilization. In addition, since the inference graph is much larger than the training subgraphs, we need to optimize off-chip communication by partitioning and scheduling algorithms. 
\begin{acks}
This work was supported by \grantsponsor{}{US NSF}{} under grant No. \grantnum{}{CNS-1643351}, \grantsponsor{}{Intel Strategic Research Alliance}{} and \grantsponsor{}{the Defense Advanced Research Projects Agency}{} under contract No. \grantnum{}{FA8750-17-C-0086}. 
\end{acks}

\clearpage
\bibliographystyle{acmart/ACM-Reference-Format}
\bibliography{citation} 

\end{document}